\documentclass[onecolumn]{IEEEtran}

\usepackage{amsfonts, amsmath, amssymb, amsthm}
\usepackage{verbatim}
\usepackage{subfigure}
\usepackage{cite}
\usepackage{hyperref}

\DeclareMathOperator*{\argmin}{arg\,min}
\usepackage{tikz}
\usetikzlibrary{arrows,automata}
\tikzstyle{state}=[circle,fill=!50,text=white,minimum size =25pt,inner sep=0pt]
\tikzstyle{selected state}=[state, fill=!100]
\tikzstyle{Markov state}=[circle,fill=!50,text=white,minimum size =40pt,inner sep=0pt]
\begin{document}

\newtheorem{Lemma}{Lemma}
\newtheorem{Theorem}{Theorem}
\newtheorem{Example}{Example}
\newtheorem{Definition}{Definition}
\newtheorem{ProblemStatement}{Problem Statement}
\newtheorem{Corollary}{Corollary}
\newtheorem{Claim}{Claim}
\newtheorem{Fact}{Fact}
\newtheorem{Property}{Property}
\newtheorem{Remark}{Remark}

\def \hence {\Rightarrow}
\def \ham {\mathcal{H}}
\def \Ham {\mathcal{H}}
\def\Sec {\S}
\newcommand \remove[1] {}
\newcommand \hide[1] {}
\newcommand \green[1]         {{\color{green}#1}}
\newcommand \blue[1]         {{\color{blue}#1}}
\newcommand \yellow[1]         {{\color{yellow}#1}}
\newcommand \red[1]         {{\color{red}#1}}

\newcommand{\boldf}{\mathbf{f}}

\newcommand{\bFu}[1]{\bF^{(#1)}}
\newcommand{\bI}{\mathbf{I}}
\newcommand{\BS}{\mathbb{S}}
\newcommand{\bS}{\mathbf{S}}
\newcommand{\bR}{\mathbf{R}}
\newcommand{\bu}{\mathbf{u}}
\newcommand{\bx}{\mathbf{x}}
\newcommand{\by}{\mathbf{y}}
\newcommand{\bF}{\mathbf{F}}

\title{Optimal Patching in Clustered Malware Epidemics}



\author{S. Eshghi, MHR. Khouzani, S. Sarkar, S. S. Venkatesh}
        \maketitle

{\renewcommand{\thefootnote}{} \footnotetext{
S. Eshghi, S. Sarkar and S. S. Venkatesh are with the Department of Electrical and Systems Engineering at the University of Pennsylvania,
Philadelphia, PA, U.S.A. Their email addresses are \emph{eshghi,swati,venkates@seas.upenn.edu}. MHR. Khouzani is with the Department of Electrical Engineering at the University of Southern California, Los Angeles, CA. His e-mail address is \emph{rezaeikh@usc.edu}.
\ \\
This paper was presented [in part] at the IEEE Information Theory and Applications Workshop (ITA '12) , San Diego, CA, February, 2012}}

\maketitle

%
%
\begin{abstract}
Studies on the propagation of malware in mobile networks have revealed that the spread of malware can be highly inhomogeneous\red{\hide{ across different regions.}}.
Platform diversity, contact list utilization by the malware, clustering in the network structure, etc. can also lead to differing spreading rates. In this paper, a general formal framework is proposed for leveraging such \hide{heterogeneity information}heterogeneity to derive optimal patching policies that attain the minimum aggregate cost due to the spread of malware and the surcharge of patching.
Using Pontryagin's Maximum Principle for a stratified epidemic model, it is analytically proven that in the mean-field deterministic regime,
optimal patch disseminations are simple single-threshold policies\hide{ that are amenable to implementation in a distributed manner}. Through numerical simulations, the behavior of optimal patching policies is investigated in sample topologies and their advantages are demonstrated.
\end{abstract}

\IEEEpeerreviewmaketitle

\section{Introduction}\label{sec:Introduction}
Worms (self-propagating malicious codes) are a decades-old
threat in the realm of the Internet. Worms undermine the network \hide{by performing various malicious activities}in various ways: they can eavesdrop {on} and analyze traversing data, access privileged information, hijack sessions, disrupt network functions such as routing, etc. 
Although the Internet is the traditional arena for \hide{malicious codes such as} trojans, spyware, and viruses, \hide{the battle is expanding to new territories: }the current boom in mobile devices, combined with their spectacular software and hardware capabilities, has created a tremendous opportunity for future malware. Mobile devices communicate with each other and with computers through myriad means\hide{\hide{. N}{; n}ot only can they interact using}-- Bluetooth or Infrared when they are in close proximity, \hide{ or through an exchange of} multimedia \hide{{content}} messages (MMS), \hide{they can {also} have ubiquitous access to} mobile Internet, and peer to peer networks\hide{ via a telecom{munications} provider}. Current smartphones are equipped with operating systems, CPUs, and memory powerful enough to execute\hide{increasingly more} complex codes. Wireless malware such as \emph{cabir}, \emph{skulls}, \emph{mosquito}, \emph{commwarrior} have already sounded the alarm~\cite{ramachandran2007stability}. It \hide{has, in fact,}has been theoretically predicted~\cite{wang2009understanding} that it is only a matter of time before major malware outbreaks are witnessed in the wireless domain.

Malware spreads when an infective node \emph{contacts}, i.e., communicates  with,  a susceptible node, i.e., a node
 without a copy of the malware and vulnerable to it. This spread can be countered through patching~\cite{zhu2009social}: the \hide{underlying }vulnerability utilized by
the worm can be fixed by installing security patches that immunize the susceptible and potentially remove the malware {from the infected, hence simultaneously healing and immunizing} 
infective nodes.
However, the distribution of these patches
burdens the limited resources of the network, and can \hide{{lead to}}\remove{significant}\hide{{major}}{wreak havoc on the system if
not carefully controlled}. In wired networks, the spread of \emph{Welchia},
a counter-worm to thwart \emph{Blaster}, \hide{created substantial  traffic which in turn }rapidly destabilized
important sections of the Internet~\cite{leyden2003welchia}. Resource constraints are even more pronounced in wireless networks, {\hide{in which}where} bandwidth is \hide{{\hide{constrained}limited} and is }more sensitive to overload
and nodes {\hide{are limited in their}have limited} energy reserves.
Recognizing the above, works such as~\cite{khouzani2010dispatch,khouzani2011optimal} have included the cost of patching in the aggregate damage of the malware and have characterized the optimal dynamic patching policies that attain desired trade-offs between the patching efficacy and the extra taxation of {\hide{the}}network resources. \hide{However, as we will explain next, these studies suffer{\hide{ed}} from a {common} drawback{\hide{:}; namely, a strong simplifying assumption.

The results in~\cite{khouzani2010dispatch,khouzani2011optimal}}}These results, however,
critically rely on the homogeneous mixing assumption: that all
pairs of nodes have identical expected inter-contact times.
Optimality may now be attained by policies that constrain
the  nodes in the same state to take the same
action. While this assumption may serve as an approximation in cases where  detailed information about the network is unavailable, studies~\cite{wang2009understanding,mickens2005modeling,ramachandran2007modeling,chen2005spatial,li2010cpmc} show that the spread of malware in mobile networks can be {\hide{considerablysignificantly}}very inhomogeneous, owing primarily to the non-uniform distribution of nodes. Thus, a uniform action may be sub-optimal.

 \hide{Wireless nodes in  high density areas, sometimes referred to as ``popular content {\hide{''} regions''} or ``hot-spots'',  have more frequent opportunities to contact each other 
than to contact nodes in distant and less dense areas.
Heterogeneity in the contact process can arise for other reasons too{\hide{.}:} 
malware may have a lower rate of {spread\hide{contact}} between devices with differing operating systems or communication protocols 
  \cite{yang2008improving,lioptimal,nguyen2009macro}. Mobile malware may also select targets from  the  address books of the infective hosts~\cite{zhu2009social}: the contact rate is therefore higher amongst \emph{friendship cliques} in the social network of users. 
Malware that spreads using (mobile or wired) Internet can have easier access to the IP-addresses of the subnet to which the infective host belongs compared to the rest of the masked IP addresses \cite{liljenstam2002mixed}. 
The behavioral pattern of the users can also cause heterogeneous contact rates, e.g., a safe user may avoid unsolicited mass-messages or may install  firewalls, hence hindering the spread of malware as compared to one with risky behavior.
Moreover, cloud-computing seems to be a natural context for heterogeneous mixing: computers inside the same cluster have a much higher speed of communication amongst themselves than with computers of distant clusters.} 
We can motivate heterogeneous epidemics in various ways: 1) Proximity-based spread, where locality causes heterogeneity;
2) Software/protocol diversity~\cite{yang2008improving,lioptimal,nguyen2009macro}, which has even been envisioned as a defense mechanism against future threats~\cite{yang2008improving}; 3) IP space diversity, e.g., worms that use the IP masks of specific Autonomous Systems (AS) to increase their yield~\cite{liljenstam2002mixed};
4) Differing clique sizes, especially for epidemics for which there is an underlying social network graph~\cite{faghani2009malware, zhu2009social}; 5) Behavioral patterns, specifically how much risky behaviour every agent engages in\cite{zyba2009defending}; 6) Clusters in cloud computing~\cite{altunay2010optimal}; 7) Technology adoption, belief-formation over social media, and health care~\cite{feichtinger1994dynamic,sethi1977dynamic,sethi2000optimal,behncke2000optimal,wickwire1977mathematical}.

Indeed, many works have proposed practical methods to identify{\hide{ and},} characterize and incorporate such inhomogeneities 
to more accurately predict the spread of infection \cite{liljenstam2002mixed,wang2009understanding,cuzick1990spatial,hsu2006capturing,ramachandran2007modeling,chen2005spatial,antunovic2011percolation}, etc. \hide{Specifically,\cite{antunovic2011percolation} considers the case of regular random graphs and characterizes the relative success of two competing infections. }Relatively few, e.g., \cite{lioptimal,li2010cpmc,zhu2009social},
  consider the cost of patching and seek to minimize it in the presence of heterogeneous contact processes.
The proposed policies in~\cite{li2010cpmc,zhu2009social} are heuristic and apply to  specific settings.
 The only paper we could find that provides \emph{provably optimal} patching policies for  heterogeneous networks is~\cite{lioptimal}. They{, however,} focus on SIS models and optimize only in the space {of} static policies (those
 that do not vary patching rates over time) therein. Patching performance can {be significantly improved} if we allow the patching rates to
vary dynamically in accordance with the evolution of the {infection\hide{contagion}}. Characterization of the optimal controls in the space of dynamic and clustered policies has, however, so far remained elusive.

\hide{To the best of our knowledge, this is the first work that consider the most general form of a stratified epidemic and provide analytical structural guarantees for a dynamic patching.
Our model is general enough to capture any clustering of the nodes with arbitrary inter-contact rates of interaction.} \hide{In our model, the patching can immunize the susceptibles
and may or may not heal the infectives. We also consider both possibilities whether or not recipients of the patch further propagate it to others.
We also assume any arbitrary convex or concave functions for the cost rate of the infective nodes. The power of our analytical results is in the extensive generality of our model.}

We propose a formal framework for  deriving \emph{dynamic optimal} patching policies that leverage
heterogeneity in the network structure to attain the minimum possible aggregate cost due to the spread of malware and the overhead of patching. We assume arbitrary (potentially non-linear) functions for the cost rates of the infective nodes. We consider both \emph{non-replicative} and \emph{replicative} patching: in the former, some of the hosts are pre-loaded with the patch, which they transmit to the rest. In the latter,  each recipient of the patch can also forward the patch to nodes that it contacts by a mechanism similar to the spread of the malware itself. In our model, patching can immunize susceptible nodes
and may or may not heal infective nodes.\hide{We also assume any arbitrary convex or concave functions for the cost rate of the infective nodes.} The framework in each case relies on  optimal control formulations
that cogently capture the effect of the
patching rate controls
 on the state dynamics and their resulting trade-offs. We accomplish this by using a combination
 of damage functions associated with the controls and  a \emph{stratified}\footnote{Known by other terms such as structured, clustered, multi-class, multi-type, multi-population, compartmental epidemic models, and sometimes loosely as heterogeneous, inhomogeneous or spatial epidemic models.} mean-field deterministic epidemic model in which nodes are divided into different types. Nodes of the same type homogeneously mix with a rate specific to that type, and nodes {of\hide{across}} different types contact each other at rates particular to {that\hide{each}} pair of types. \hide{If two types  do not interact,  the corresponding inter-contact rates are set to zero. }The model
 can therefore  capture any communication topology between different groups of nodes. Above and beyond, it
  can exploit  the inhomogeneity in the network
to enable  a better utilization of the resources. Such higher patching efficacy is achieved by allowing the patching controls to depend on node types, which in turn leads to  multidimensional (dynamic) optimal control
formulations. We first develop our system dynamics and objectives (\S\ref{sec:Sys_Model}) and characterize optimal non-replicative (\S\ref{sec:Optimal_NONrep}) and replicative (\S\ref{sec:Optimal_rep}) patching. We then analyze an alternate objective (\S\ref{sec:Generalizations}) and present numerical simulation of our results (\S\ref{sec:Numericals}).

 Multidimensional optimal control formulations, particularly those in the solution space of functions rather than variables, are usually associated with the pitfall of amplifying  the complexity of the optimization.  An important
 contribution of the paper, therefore, is to prove that for both non-replicative and replicative settings the optimal control associated with each type
has a simple structure provided the corresponding patching cost   is either concave or convex.
Furthermore, our analysis, using Pontryagin's Maximum Principle, reveals that
 the structure of the optimal control for a specific type depends only on the nature of
the corresponding patching cost  and not on those of
other types. This holds even though the control for each type affects immunization
and healing in other types and the spread of the infection in general. Specifically, if the patching cost associated with the control for
a given type is concave, irrespective of the nature of
the patching costs for other types,  the corresponding optimal control turns out to be a bang-bang function with at most
one jump: up to a certain threshold time (possibly different for different types) it selects the maximum possible patching rate and subsequently
it stops patching altogether. \hide{The thresholds will be different for different types
and may now be computed  using efficient off-the-shelf  numerical techniques. }If the patching cost is strictly convex, the decrease from the maximum
to the minimum patching rate is continuous rather than abrupt, and monotonous. To the best of our
knowledge, such simple structure results have not been established in the context of (static or dynamic) control of
heterogeneous epidemics. Our numerical calculations reveal  a series of interesting behaviors of optimal patching policies for different sample topologies.

\red{\hide{As a final comment, although we focus  on malware propagation, stratified or clustered epidemics can capture
state evolutions  in a diverse set of applications, such as technology adoption,
belief propagation over social media, and health care.
The framework that we propose may be used to attain desired tradeoffs between
the resources consumed in applying the control and the damage induced
by the proliferation of undesirable states in these contexts.  This is again achieved by exploiting  heterogeneities
inherent to such systems
through multi-dimensional dynamic optimal control.
Our work advances the state of the art of dynamic optimal control in these settings by considering a model that captures any degree of inhomogeneity.
Existing research in these areas either considers a homogeneous epidemic (e.g. \cite{sethi2000optimal,behncke2000optimal}),
or investigates the optimal control only numerically for a limited number of distinct clusters and guarantees no structural results (e.g.\cite{ndeffo2010optimization,rowthorn2009optimal}).}}

\section{System Model and Objective Formulation}\label{sec:Sys_Model}
In this section we describe and develop  the model of the state dynamics of the system as a general \emph{stratified} epidemic for both non-replicative (\Sec\ref{subsec:non_rep_model}) and replicative (\Sec\ref{subsec:rep_model}) patching, motivate the model (\Sec\ref{subsec:motivation}), formulate the aggregate cost of patching, and cast this resource-aware patching as a multi-dimensional optimal control problem (\Sec\ref{subsec:Cost}). This formulation relies on a key property of the state dynamics which we isolate in \Sec\ref{subsec:Observations}. We develop solutions in this model framework and present our main results in~\Sec\ref{sec:Optimal_NONrep} and~\Sec\ref{sec:Optimal_rep}.

Our models are based on mean-field limits of Poisson contact processes for which pathwise convergence results have been shown (c.f.~\cite[p.1]{kurtz1970solutions},~\cite{gast2010mean}).

\subsection{Dynamics of non-replicative patching}\label{subsec:non_rep_model}
A node is \textbf{infective} if it has been contaminated by  the malware, \textbf{susceptible}
if it is vulnerable to the infection but not yet infected, and  \textbf{recovered}
if it is immune to the malware.
An infective node spreads the malware to a susceptible one while transmitting  data or control messages.
The network consists of nodes that can be stratified into $M$ different \emph{types} (equivalently, \emph{clusters, segments, populations, categories, classes, strata}). The population of these types need not be equal.
Nodes of type $i$ contacts those of type $j$ at rate $\beta_{ij}$. 

For type $i$, $S_i(t)$, $I_i(t)$, and $R_i(t)$ are respectively the fraction of susceptible, infective and recovered states at time $t$. Therefore, for all $t$ and all $i$, we have $S_i(t)+ I_i(t) + R_i(t) = 1$. We assume that during the course of the epidemic, the population of each type is stable and does not change with time.

Amongst each type, a pre-determined set of nodes, called \emph{dispatchers}, are preloaded with the appropriate patch. Dispatchers can transmit patches to both susceptible and infective nodes, \emph{immunizing} the susceptible
and possibly \emph{healing} the infective; in either case successful transmission converts the target node to the recovered state.
In \emph{non-replicative} patching (as opposed to \emph{replicative} patching- see~\Sec\ref{subsec:rep_model}) the recipient nodes of the patch do not propagate it any further.\footnote{This may be preferred if the
patches themselves {can} be contaminated and cannot be reliably authenticated.}
Dispatchers of type $i$ contact nodes of type $j$ at rate $\bar\beta_{ij}$, which may be different from the malware contact rate $\beta_{ij}$ between these two types.
Examples where contact rates may be different include settings where the network manager may utilize a higher priority option for the distribution of patches, ones where the malware utilizes legally restricted means of propagation not available to dispatchers, or ones where the patch is not applicable to all types, with the relevant $\bar\beta_{ij}$ now being zero. The fraction of dispatchers in type $i$, which is fixed over time in the non-replicative setting, is $R_i^0$, where $0 \leq R_i^0 < 1$.

Place the time origin $t=0$ at the earliest moment the infection is detected and the appropriate patches generated. Suppose that at $t=0$, for each $i$, an initial fraction $0 \leq I_i(0) = I_i^0 \leq 1$ of nodes of type $i$ are infected; we set $I_i^0 = 0$ if the infection does not initially exist amongst a type $i$. At the onset of the infection, the dispatchers are the only agents immune to the malware, hence constituting the initial population of recovered nodes. We therefore identify $R_i(0)=R_i^0$. In view of node conservation, it follows that $S_i^0 = 1 - I_i^0 - R_i^0$ represents the initial fraction $S_i(0)$ of susceptible nodes of type $i$.

At any given $t$, susceptibles of type $i$ may be contacted by infectives of type $j$ at rate  $\beta_{ji}$. We may fold resistance to infection into the contact rates and so from a modeling perspective, we may assume that susceptibles contacted by infectious agents are instantaneously infected. Accordingly, susceptibles of type $i$ are transformed to infectives (of the same type) at an aggregate rate of $S_i(t)\sum_{j=1}^M\beta_{ji}I_j(t)$ by contact with infectives of any type.

The system manager regulates the resources consumed in the patch distribution by dynamically controlling the rate at which dispatchers contact susceptible and infective nodes. For each $j$, let the control function $u_j(t)$ represent the rate of transmission attempts of dispatchers of type $j$ at time $t$. We suppose that the controls are non-negative and bounded,
\begin{equation} \label{Controller_Constraint}
 0\leq  u_j(\cdot) \leq u_{j,\max}.
\end{equation}
We will restrict consideration to control functions $u_j(\cdot)$ that have a finite number of points of discontinuity. We say that a control (vector) $\mathbf{u}(t) = \bigl(u_1(t),\dots,u_M(t)\bigr)$ is \emph{admissible} if each $u_j(t)$ has a finite number of points of discontinuity.

Given the control $\mathbf{u}(t)$, susceptibles of type $i$ are transformed to recovered nodes of the same type at an aggregate rate of $S_i(t)\sum_{j=1}^M\bar\beta_{ji} R_j^0 u_j(t)$ by contact with dispatchers of any type. A subtlety in the setting is that the dispatcher may find that the efficacy of the patch is lower when treating infective nodes. This may model situations, for instance, where the malware attempts to prevent the reception or installation of the patch in an infective host, or the patch is designed only to remove the vulnerability that leaves nodes exposed to the malware but does not remove the malware itself if the node is already infected. We capture such possibilities by introducing a (type-dependent) coefficient $0\leq\pi_{ji}\leq1$ which represents the efficacy of patching an infective node: $\pi_{ji}=0$ represents one extreme where a dispatcher of type $j$ can only immunize susceptibles but can not heal infectives of type $i$, while $\pi_{ji}=1$ represents the other extreme where contact with a dispatcher of type $j$ both immunizes and heals nodes of type $i$ equally well; we also allow $\pi_{ij}$ to assume intermediate values between the above extremes.
An infective node transforms to the recovered state if a patch heals it; otherwise, it remains an infective. Infective nodes of type $i$ accordingly recover at an aggregate rate of
$I_i(t)\sum_{j=1}^M\pi_{ji}\bar\beta_{ji}R_j^0u_j(t)$ by contact with dispatchers.

We say that a type $j$ is a \emph{neighbour} of a type $i$ if $\beta_{ij}>0$, and $S_j>0$ (i.e., infected nodes of type $i$ can contact nodes of type $j$). There is now a natural notion of a topology that is inherited from these rates with types as vertices and edges between neighboring types. Figure~\ref{fig:topologies} illustrates some simple topologies.
\begin{figure*}[htbp]
\centering
\subfigure{
\includegraphics[scale=0.4]{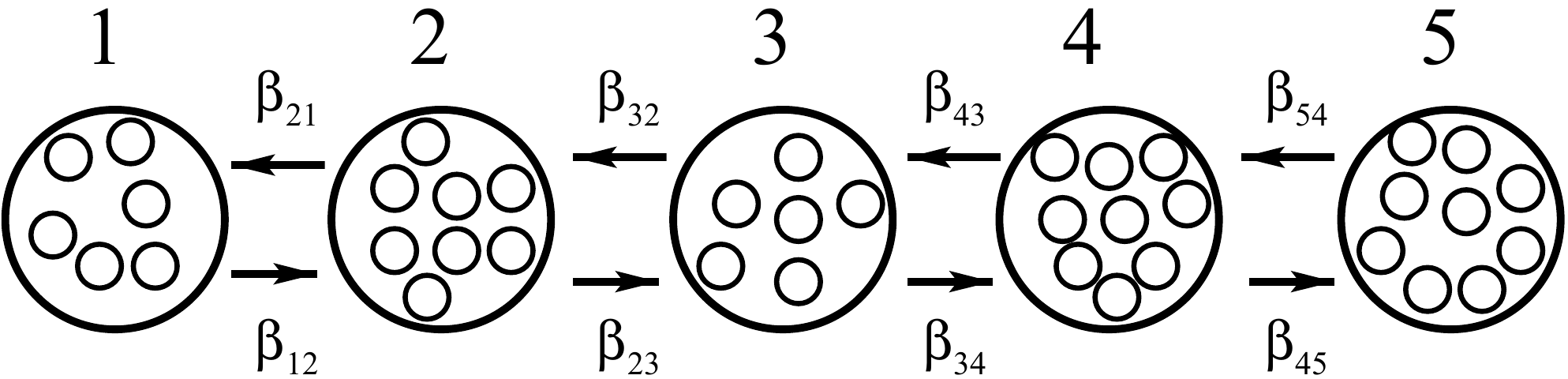}
\label{subfig:linear}
}
\subfigure{
\includegraphics[scale=0.4]{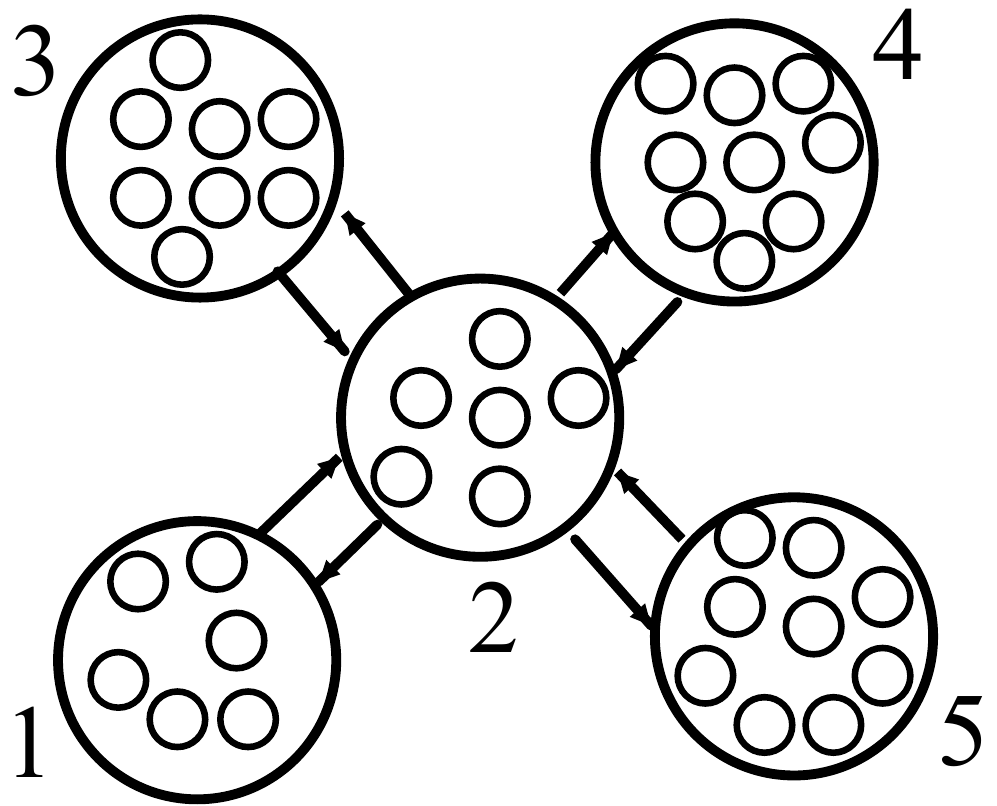}
\label{subfig:star}
}
\subfigure
{
\includegraphics[scale=0.4]{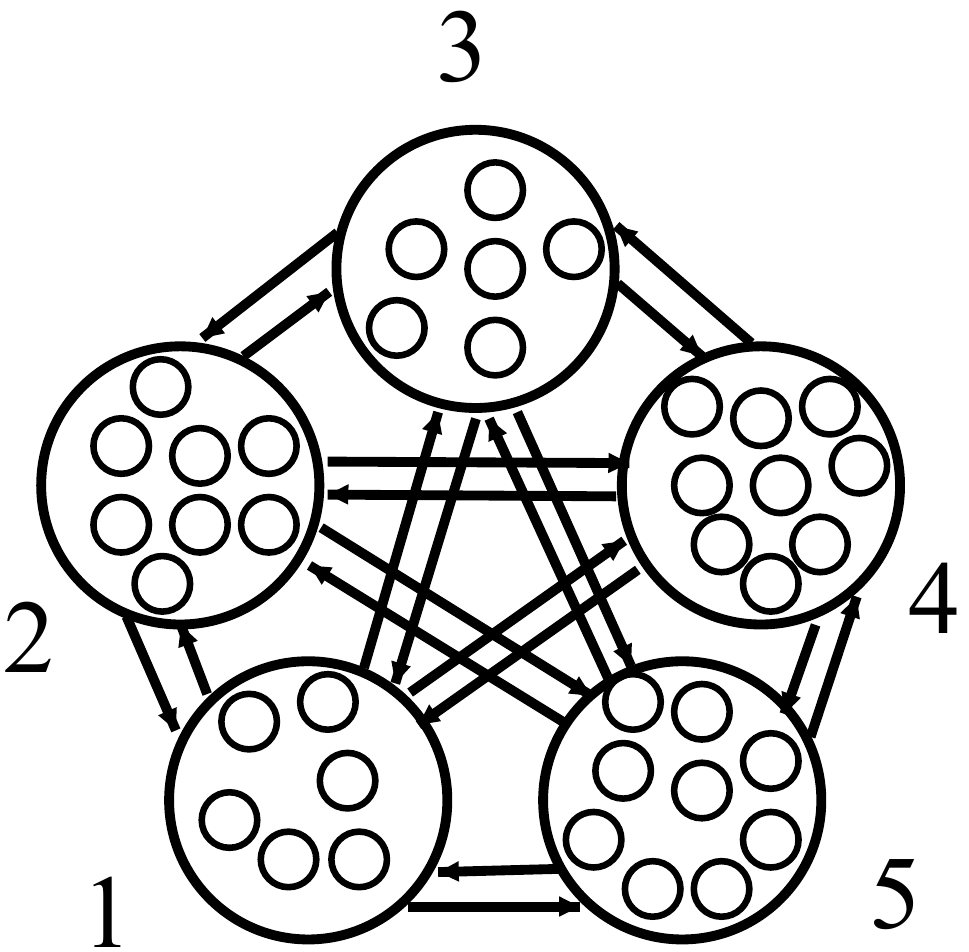}
\label{subfig:complete}
}
\label{fig:subfigureExample}
\caption{Three sample topologies of 5 hotspot regions: linear, star and complete. For instance, nodes of hotspot 1 in the linear topology can only communicate with nodes of hotspots 1 and 2: they contact nodes of
hotspot 1 at rate $\beta_{11}$ and nodes of hotspot 2 at rate $\beta_{12}$. }\label{fig:topologies}
\end{figure*}
For a given topology inherited from the rates $\{\beta_{ij},1\leq i,j\leq M\}$ there is now another natural notion, that of \emph{connectivity}: we say that type $j$ is connected to type $i$ if, for some $k$, there exists a sequence of types $i = s_1\mapsto s_2 \mapsto \dots\mapsto s_{k-1}\mapsto s_k = j$ where type $s_{l+1}$ is a neighbour of type $s_l$ for $1\leq l<k$.  We assume that each type is either initially infected ($I_i(0)>0$), or is connected to an initially infected type. We also assume that for every type $i$ such that $R_i^0>0$, there exists a type $j$ for which $\bar{\beta}_{ij}>0$, i.e., type $i$ can immunize nodes of at least one type, and there exist types $k$ and $l$ for which $\beta_{ki}>0$ and $\beta_{il}>0$, $S_l(0)>0$, i.e., the infection can spread to and from that type. (In most settings we may expect, naturally, that $\beta_{ii}>0$ and $\bar\beta_{ii}>0$.)

Thus, we have\hide{ the following system of differential equations}:\footnote{We use dots to denote \emph{time} derivatives throughout, e.g., $\dot{S}(t) = dS(t)/dt$.}
\begin{subequations}\label{Imm_Sys_NonRep}
\begin{align}
\dot{S}_i
&= -\sum_{j=1}^{M}\beta_{ji}I_jS_i - S_i\sum_{j=1}^M\bar\beta_{ji} R_j^0u_j, \allowdisplaybreaks\\
\dot{I}_i &= \sum_{j=1}^{M}\beta_{ji}I_jS_i -I_i\sum_{j=1}^M\pi_{ji} \bar\beta_{ji}R_j^{0}u_j, 
\end{align}
\end{subequations}
where, by writing $\mathbf{S}(t) = \bigl(S_1(t),\dots,S_M(t)\bigr)$, $\mathbf{I}(t) = \bigl(I_1(t),\dots,I_M(t)\bigr)$, and $\mathbf{R}^0(t) = \bigl(R^0_1(t),\dots, R^0_M(t)\bigr)$ in a compact vector notation, the initial conditions and state constraints are given by
\begin{gather}
\mathbf{S}(0)=\mathbf{S}^0\succeq\mathbf{0},\quad
\mathbf{I}(0)=\mathbf{I}^0\succeq\mathbf{0}, \label{Imm_Sys_NonRep_initialstate}\\
\mathbf{S}(t)\succeq\mathbf{0}, \quad \mathbf{I}(t)\succeq\mathbf{0},
\quad \mathbf{S}(t) + \mathbf{I}(t) \preceq \mathbf{1}-\mathbf{R}^0.
\label{Imm_Sys_NonRep_stateconstraint}
\end{gather}
In these expressions $\mathbf{0}$ and $\mathbf{1}$ represent vectors all of whose components are $0$ and $1$ respectively, and the vector inequalities are to be interpreted as component-wise inequalities. Note that the evolution of $\mathbf{R}(t)$ need not be explicitly considered since at any given time, node conservation gives $R_i(t)=1-S_i(t)-I_i(t)$. We henceforth drop the dependence on $t$ and make it implicit whenever we can do so without ambiguity. \hide{The initial conditions $\mathbf{S}^0\succ\mathbf{0}$ and $\mathbf{I}^0\succeq\mathbf{0}$ given in~\eqref{Imm_Sys_NonRep_initialstate} require that  each type has a non-vanishing fraction of susceptible agents initially, though we permit some of the types to be initially infection-free. This allows us to examine the propagation of infection and the evolution of the optimal controls in initially uninfected regions.Furthermore, WLoG we can consider $u_{i,\max}=1$, as any scaling can be applied to the relevant $\bar\beta_{ji}^{(N)}$s.\footnote{In general, if the network manager cannot control a certain type $i$, the corresponding $\bar{beta_{ji}}$'s in the model will be zero} }

\subsection{Dynamics of replicative patching} \label{subsec:rep_model}

In the replicative setting, a recipient of the patch can forward it to other nodes upon subsequent contact. Thus, recovered nodes of type $i$ are added to the pool of dispatchers of type $i$, whence the fraction of dispatchers of type $i$ grows from the initial $R_i(0) = R_i^0$ to $R_i(t)$ at time $t$. This should be contrasted with the non-replicative model in which the fraction of dispatchers of type $i$ is fixed at $R_i^0$ for all $t$.

The system dynamics equations given in~\eqref{Imm_Sys_NonRep} for the non-replicative setting now need to be
modified to take into account the growing pool of dispatchers. 
While in the non-replicative case we chose the pair $\bigl(\mathbf{S}(t), \mathbf{I}(t)\bigr)$ to represent the system state, in the replicative case it is slightly more convenient to represent the system state by the explicit triple $\bigl(\mathbf{S}(t), \mathbf{I}(t), \mathbf{R}(t)\bigr)$. The system dynamics are now governed by:
\begin{subequations}\label{Imm_Sys_Rep}
\begin{align}
\dot{S}_i
&= -\sum_{j=1}^{M}\beta_{ji}I_jS_i  - S_i\sum_{j=1}^M\bar\beta_{ji} R_ju_j, \allowdisplaybreaks\\
\dot{I}_i &= \sum_{j=1}^{M}\beta_{ji}I_jS_i -I_i\sum_{j=1}^M\pi_{ji} \bar\beta_{ji}R_{j}u_j, \allowdisplaybreaks \\
\dot{R}_i&=S_i\sum_{j=1}^M\bar\beta_{ji} R_ju_j+ I_i\sum_{j=1}^M\pi_{ji}\bar\beta_{ji}R_{j}u_j,
\end{align}
\end{subequations}
with initial conditions and state constraints given by
\begin{gather}
\mathbf{S}(0)=\mathbf{S}^0\succeq\mathbf{0},\quad
\mathbf{I}(0)=\mathbf{I}^0\succeq\mathbf{0},\quad
\mathbf{R}(0)=\mathbf{R}^0\succeq\mathbf{0},\label{Imm_Sys_Rep_init}\\
 \mathbf{S}(t)\succeq\mathbf{0}, \quad \mathbf{I}(t)\succeq\mathbf{0},
 \quad \mathbf{R}(t)\succeq\mathbf{0}, \quad  \mathbf{S}(t)+\mathbf{I}(t)+\mathbf{R}(t)=\mathbf{1}.\label{Imm_Sys_Rep_stateconstraint}
\end{gather}
The assumptions on controls and connectivity are as in \Sec\ref{subsec:non_rep_model}.

\subsection{Motivation of the models and instantiations} \label{subsec:motivation}

We now motivate the stratified epidemic models~\eqref{Imm_Sys_NonRep} and~\eqref{Imm_Sys_Rep} through
different examples, instantiating the different types in each context.

\subsubsection{Proximity-based spread---heterogeneity through locality}
The overall roaming area of the nodes can be divided into regions (e.g.,  hotspots, office/residential areas, central/peripheral areas, etc.) of different densities (fig.~\ref{fig:topologies}). One can therefore stratify the nodes  based on their locality, i.e., each type corresponds to a region. IP eavesdropping techniques (using software such as \verb|AirJack|, \verb|Ethereal|, \verb|FakeAP|, \verb|Kismet|, etc.)
 allow malware  to  detect new victims
 in the vicinity of the host. Distant nodes have more attenuated signal strength (i.e., lower SINR) and are therefore less likely to be detected.  Accordingly, malware (and also patch) propagation rates $\beta_{ij}$ (respectively $\bar\beta_{ij}$) are related to the local densities of the nodes in each region and decay with an increase in the distance
between regions $i$ and $j$: typically $\beta_{ii}$ exceeds $\beta_{ij}$ for $i \neq j$, likewise
for  $\bar\beta_{ij}.$  The same phenomenon was observed for
 malware such as \verb|cabir| and \verb|lasco| that use Bluetooth and Infrared to propagate.


\subsubsection{Heterogeneity through software/protocol diversity}
A network that relies on a homogeneous software/protocol is vulnerable to an  attack that exploits a common weakness  (e.g., a buffer overflow vulnerability).
Accordingly, inspired by the natural observation that \emph{the chances of survival are improved by heterogeneity}, increasing the network's heterogeneity without sacrificing interoperability {\hide{is}has been} proposed as a defense mechanism~\cite{yang2008improving}. In practice, mobile nodes use different operating systems and communication protocols, e.g., \verb|Symbian|, \verb|Android|, \verb|IOS|, \verb|RIM|, \verb|webOS|, etc. Such heterogeneities lead to dissimilar rates of propagation of malware amongst different types, where each type represents a specific OS, platform, software, protocol, etc. In the extreme case, the malware may not  be able to contaminate nodes of certain types. The patching response should take such
inhomogeneities into account in order to optimally utilize {\hide{the}}network resources, since the rate of patching {\hide{might}can} also be dissimilar among different types.

\subsubsection{Heterogeneity through available IP space}
Smartphone trojans like \verb|skulls| and \verb|mosquito| spread using Internet or P2P networks. 
In such cases the network can be decomposed into \emph{autonomous systems} (ASs) {with each type representing an AS \cite{liljenstam2002mixed}}. 
A worm  either scans IP addresses uniformly randomly or uses the IP masks of ASs to restrict its search domain and increase its rate of finding new susceptible nodes. In each of these cases the contact rates differ between different AS's depending on the actual number of assigned IPs in each IP sub-domain and the maximum size of that IP sub-domain.

\subsubsection{Heterogeneity through {\hide{size of the cliques}differing clique sizes}}
Malware that specifically spreads in social networks has been recorded in the past few years~\cite{faghani2009malware}. Examples include \verb|Samy| in MySpace in 2005 and \verb|Koobface| in MySpace and Facebook in 2008. \verb|Koobface|, for instance, spread by delivering (contaminated) messages to the ``friends'' of an infective user.
MMS based malware such as \verb|commwarrior|{\hide{may} can} also utilize the contact list of an infective host to access new handsets.
In such cases the social network graph can be approximated by a collection of friendship \emph{cliques}.\footnote{A clique is a maximal complete subgraph \hide{(a collection of nodes and all \hide{(most of)} the possible links between them)}of a graph\cite[p.~112]{bollobás1998modern}.} Users of the same clique can be regarded as the same type with the rate of contact within cliques and across cliques differing depending on the relative sizes of the cliques.
\remove{\subsubsection{heterogeneity through behavioral patterns}
Users may be stratified based on their security-consciousness for example into safe or risky types  based on usage history \cite{zyba2009defending}. Security-savvy users may use more secure protocols, avoid executing untrusted codes or mass forwarding a received message, install firewalls, disable unused features of their handsets, etc. The rate of propagation of the malware is therefore the lowest amongst safe users, higher  between safe and risky users, and
the maximum  amongst the risky users.}

\subsubsection{Cloud-computing---heterogeneity through cluster sizes}
In cluster (or grid, or volunteer) computing~\cite{altunay2010optimal}, each cluster of CPUs in the cloud constitutes a type. Any two computers in the same cluster can communicate  at faster rates than those  in different clusters. These contact rates  depend on the
communication capacity of connecting lines as well as the relative number of computers in each cluster.

\subsubsection{Clustered epidemics in technology adoption, belief-formation over social media and health care}
We now elaborate on the application of our clustered epidemics model
 in these diverse set of  contexts. First consider a rivalry between two technologies
or companies for adoption in a given population, e.g., Android and iPhone, or cable and satellite television.
Individuals who are yet to choose either may be considered as susceptibles and
those who have chosen one or the other technology would be classified as either infective or recovered depending
upon their choice. Dispatchers constitute the promoters
of a given technology (the one whose subscribers are denoted as recovered). Awareness about the technology and subsequent subscription to either may spread through
social contact between infectives and susceptibles (infection propagation
in our terminology),  and dispatchers and the rest (patching in our terminology).
Immunization of a susceptible corresponds to her adoption of the corresponding technology, while healing
of an infective corresponds to an alteration in her original choice. The stratifications may be based
on location or social cliques, and the control $\mathbf{u}$ would represent promotion efforts,
which would be judiciously selected by the proponent of the corresponding technology. Patching
may either be replicative or non-replicative depending on whether the newly subscribed users are enticed to
attract more subscribers by referral rewards.  Similarly, clustered epidemics may be used to model belief management over social media, where infective and recovered nodes represent
individuals who have conflicting persuasions and susceptibles represent those who are yet to subscribe to either doctrine.
Last, but not least, the susceptible-infective-recovered classification and immunization/healing/infection have natural connotations in the context of a biological epidemic. Here, the dispatchers correspond to health-workers who administer vaccines and/or hospitalization and the stratification is based on location. Note that in this context, patching can only be non-replicative.

\subsection{Key observations}\label{subsec:Observations}

A natural but important observation is that if the initial conditions are non-negative, then the system dynamics~\eqref{Imm_Sys_NonRep} and~\eqref{Imm_Sys_Rep} yield unique states satisfying the positivity and normalisation constraints~\eqref{Imm_Sys_NonRep_stateconstraint} and~\eqref{Imm_Sys_Rep_stateconstraint}, respectively. The proof is technical and is not needed elsewhere in the paper; we relegate it accordingly to the appendix.

\begin{Theorem}\label{thm:constraints}
The dynamical system~\eqref{Imm_Sys_NonRep} (respectively~\eqref{Imm_Sys_Rep}) with initial conditions~\eqref{Imm_Sys_NonRep_initialstate} (respectively,~\eqref{Imm_Sys_Rep_init}) has a unique state solution \big($\bS(t), \bI(t)$\big) (respectively \big($\bS(t), \bI(t), \bR(t)$\big)) which satisfies the state constraints~\eqref{Imm_Sys_NonRep_stateconstraint} (respectively,~\eqref{Imm_Sys_Rep_stateconstraint}).  For all $t>0$, $I_i(t) > 0$; $R_i(t)>0$ if $R_i(0)>0$; and $S_j(t)>0$ if and only if $S_j(0)\neq0$. For each $j$ such that $S_j(0) = 0$,  $S_j(t') = 0$ for all $t' \in (0, T]$.
\end{Theorem}
 
\hide{\begin{Lemma}\label{Lem:Positive_States:Imm}
Suppose $\mathbf{u}(\cdot)$ is any admissible control\hide{ satisfying~\eqref{Controller_Constraint}} and consider the evolution of the dynamical system~\eqref{Imm_Sys_NonRep} with initial conditions~\eqref{Imm_Sys_NonRep_initialstate} or the dynamical system~\eqref{Imm_Sys_Rep} with initial conditions~\eqref{Imm_Sys_Rep_init}. Then the state constraints~\eqref{Imm_Sys_NonRep_stateconstraint} and \eqref{Imm_Sys_Rep_stateconstraint} hold in the non-replicative and replicative cases, respectively. \hide{In both cases, moreover, the strict positivity conditions $\mathbf{S}(t)\succ\mathbf{0}$, $\mathbf{I}(t)\succ\mathbf{0}$, and $\mathbf{R}(t)\succ\mathbf{0}$ are satisfied for all $t>0$.}
\end{Lemma}}

\subsection{The optimality objective\hide{ Function}}\label{sec:objective}
\label{subsec:Cost} The network seeks to minimize the overall {\hide{combined}} cost of infection and the resource overhead of patching in a given operation time window $[0, T]$. At any given time $t$, the system incurs costs at a rate $f\bigl(\mathbf{I}(t)\bigr)$ due to the malicious activities of the malware. For instance, the malware may use infected hosts to eavesdrop, analyze,
misroute, alter, or destroy the traffic that the hosts  generate or relay. We suppose that $f(\mathbf{0}) = 0$ and make the natural assumption that the scalar function $f(\mathbf{I})$ is increasing and differentiable with respect to each $I_i$. The simplest natural candidate for $f(\mathbf{I})$ is of the form $\sum_{i=1}^M f_i(I_i)$; in this setting each $f_i$ is a non-decreasing function of its argument representing the cost of infection for type $i$ which is in turn determined by the criticality of the data of that type and its function.\footnote{Such differences themselves may be a source of stratification. In general, different types need not  exclusively reflect disparate mixing rates.} The network also benefits at the rate of $L\bigl(\mathbf{R}(t)\bigr)$, i.e., incurs a cost at the rate of $-L\bigl(\mathbf{R}(t)\bigr)$, due to the removal of {uncertainty} about the state of {the} nodes being patched. We suppose that the scalar function $L(\mathbf{R})$ is non-decreasing and differentiable with respect to each $R_i$. \hide{This cost framework also allows us to recover the results of optimal epidemic forwarding policies in multi-type DTN networks~\cite{de2010optimal,altman2009optimal} as a minor variant of our model.\footnote{\label{footnotelabel}In a Delay Tolerant Network a server may seek to broadcast a message to as many nodes as possible before a deadline by employing minimal resources such as energy and bandwidth.  In this case, susceptibles are nodes that are yet to receive the message, recovered nodes are those that have received it, and infectives are absent. The dissemination of the message may either be performed in a non-replicative or replicative manner. The system reward increases with an increase in  the number of nodes that have received  the disseminated message. In addition, the sooner the message is disseminated, the better, hence the integration of $-L(\mathbf{R}(t))$ over time (the negative sign arises because the optimization is cast as a minimization problem). In~\cite[appendix-A]{small2003shared} the integral over time of the fraction of the recipient nodes is related to the probability that a message is delivered to sink nodes before the deadline $T$. The minimum delay problem is {\hide{transfer}hence replaced by} the maximization of an expression of the form $\int_{0}^{T}\sum_{i=1}^M N_iR_i(t)\,dt$. This formulation corresponds to the special case of linear $L(\mathbf{x})=\sum_{i=1}^M c_ix_i$ in our setting with an appropriate scaling. This observation was used in~\cite{de2010optimal,altman2009optimal} to propose resource-optimal epidemic forwarding  policies in multi-type DTNs. The DTN example is captured in our model if we remove all terms from our system dynamics equations (\ref{Imm_Sys_NonRep} and \ref{Imm_Sys_Rep}) that involve $I_i$'s. This connection is also recognized in~\cite{khouzani2011optimal}.}}

In addition to the cost of infection, each dispatcher burdens the network with a cost by consuming either available bandwidth, energy reserves of the nodes (e.g., in communication and computing networks),
 or money (e.g., in technology adoption, propaganda, health-care) to disseminate the patches.
\hide{We will capture this portion of the cost \hide{with a fairly general combination of cost terms for two patching scenarios.
In the first simple scenario, the dispatchers} with a \emph{broadcast} cost for the patch.}
Suppose dispatchers of type $i$ incur cost at a rate of $R_i^0h_i(u_i)$.\hide{In the alternative scenario, dispatchers may { only transmit\hide{only}} to the nodes that have not yet received the patch.\footnote{This can be achieved by keeping a common database of nodes that have successfully received the patch, or {by implementing} a turn-taking algorithm preventing double targeting.
This choice of policy can remove some unnecessary transmissions of the patches and hence save on the patching overhead, but it should be immediately clear that its implementation involves some extra effort.
Note that we naturally assume that the network does not know a priori with certainty which nodes are infective, and hence it cannot differentiate between susceptible and infective nodes. Consequently, even when $\pi_{ij}=0$, i.e., the system manager knows that the patch cannot remove the infection and only immunizes the susceptible, still the best it may be able to do is to forward the message to any node that has not previously received it.}
Hence, the cost of patching in this case can in general be represented by a combination of terms of the following structure: $\sum_{i=1}^M \sum_{j=1}^M R_i^0\bar\beta_{ij}(S_j+I_j)h_i^2(u_i)$.}
We suppose that the overhead of extra resource (bandwidth or energy or money)  consumption at time $t$ is then given by a sum of the form $\sum_{i=1}^MR_i^0h_i(u_i)$. The scalar functions $h_i(\cdot)$ represent how much resource is consumed for transmission of the patch by nodes of each type and how significant this extra taxation of resources is for each type. Naturally enough, we assume these functions are non-decreasing and satisfy $h_i(0)=0$ and $h_i(\gamma)>0$ for $\gamma>0$. We assume, additionally, that each $h_i$ is twice differentiable. Following the same lines of reasoning, the corresponding expression for the cost of replicative patching is of the form $\sum_{i=1}^MR_ih_i(u_i)$.

With the arguments stated above, the aggregate cost  for {\hide{the}}non-replicative patching is given by an expression of the form
\begin{equation}\label{Imm_cost_NonRep}
 \begin{split}
 J_{non-rep}=\int_{0}^T\bigg(f(\mathbf{I})-L(\mathbf{R})+\sum_{i=1}^MR_i^0 h_i(u_i) \bigg)\,dt,
\hide{+\sum_{i=1}^M \sum_{j=1}^M \bar\beta_{ij}R_i^0(S_j+I_j)h_i^2(u_i)}
\end{split}
\end{equation}
while for \hide{the}replicative patching, the aggregate cost is of the form
\begin{equation}\label{Imm_cost_Rep}
 \begin{split}
 J_{rep}=\int_{0}^T\bigg(f(\mathbf{I})-L(\mathbf{R})+\sum_{i=1}^M R_ih_i(u_i) \bigg)\,dt.
\hide{+\sum_{i=1}^M \sum_{j=1}^M \bar\beta_{ij}R_i(S_j+I_j)h_i^2(u_i)}
\end{split}
\end{equation}
\hide{In the interests of simplicity, we will call the integrands of the above equations $\xi_0$ and $\xi_1$ respectively.In the generalization section, we will look at an alternative cost model which emphasizes the cost of reception of the patches as well as their dissemination.}

\underline{\textbf{Problem Statement}:}
The system seeks to minimize the aggregate cost {\textbf{(\underline{A})}}~in~\eqref{Imm_cost_NonRep} for non-replicative patching~(\ref{Imm_Sys_NonRep}, \ref{Imm_Sys_NonRep_initialstate}) and {\textbf{(\underline{B})}}~in~\eqref{Imm_cost_Rep} for replicative patching~(\ref{Imm_Sys_Rep}, \ref{Imm_Sys_Rep_init}) by an appropriate selection of an optimal admissible control $\mathbf{u}(t)$. \hide{,  allowing  the states to evolve {\textbf{(\underline{I})}}~as per~\eqref{Imm_Sys_NonRep} for
non-replicative, and (\underline{\textbf{II}})~as per~\eqref{Imm_Sys_Rep} for replicative patching, {\hide{and satisfy}while satisfying}
the respective initial state
conditions in \eqref{Imm_Sys_NonRep_init} and~\eqref{Imm_Sys_Rep_init}.}
\ \\
In this setting, it is clear that by a scaling of $\bar\beta_{ji}$ and $h_j(\cdot)$ we can simplify our admissibility conditions on the optimal controls by supposing that $u_{j,\max}=1$. 

It is worth noting that any control in the non-replicative case can be emulated in the replicative setting: this is because the fraction of the dispatchers in the replicative setting is non-decreasing, hence at any time instance, a feasible $u^{rep}(t)$ can be selected such that $R_i(t)u_i^{rep}(t)$ is equal\hide{ed} to $R^0_iu_i^{non-rep}(t)$. \hide{This analytic argument proves that the optimal solution in the replicative setting is more efficient that the corresponding non-replicative case.} This means that the minimum cost of replicative patching is always less than the minimum cost of its non-replicative counterpart. \hide{This is indeed compatible with our numerical results.}Our numerical results will show that this improvement is substantial. However, \hide{this added efficiency is achieved at the cost of introducing the potential risk of compromised patches}replicative patches increase the risk of patch contamination:\hide{this "hole" in the protocol can be utilized by a malware to disguise itself as a patch and propagate, in which case} the security of a smaller set of known dispatchers is easier to manage than that of a growing set whose identities may be ambiguous.
Hence, in a nutshell, if there is a dependable mechanism for authenticating patches, replicative patching is the preferred method, otherwise one needs to evaluate the trade-off between {the risk of compromised patches}  and the efficiency of the patching.

\hide{\begin{Remark}\label{Remark:DTN}
\emph{
In a Delay Tolerant Network, a server may seek to broadcast a message to as
many nodes as possible before a deadline by employing minimal resources such as energy and bandwidth.  In this case, susceptibles are the nodes that are
yet to receive the message, and the recovered are the ones that have received it. Dissemination of the message may either be performed in non-replicative
or replicative 
manner.
Infective nodes are absent in this probLemma
There is reward associated with increasing the total number of nodes that have received a copy of the
disseminated message. Also, the sooner the
message is disseminated, the better, hence the
integration of $-L(\mathbf{R}(t))$ over time (the negative sign in there because the optimization is cast as a minimization problem).
The probability that a
message is delivered to sink nodes before deadline $T$ is increasing in the integral over time of the
fraction of the recipient nodes \cite[appendix-A]{small2003shared}. Hence the minimum delay problem is transformed to
maximization of $\int_{0}^{T}\sum_{i=1}^M N_iR_i(t)\,dt$, which
corresponds to the special case of linear $L(\mathbf{x})=\sum_{i=1}^M c_ix_i$ in our setting for appropriate scalings.
This observation is used in papers such as~\cite{de2010optimal,altman2009optimal} to propose resource-optimal epidemic forwarding  policies in multi-type DTNs. Our model hence captures the systems discussed in~\cite{de2010optimal,altman2009optimal} as a special case by assuming $I_0=0$ and $f(\mathbf{I})\equiv \mathbf{0}$.
%
This connection is also recognized in~\cite{khouzani2011optimal}.}
\end{Remark}}

\section{Optimal Non-Replicative Patching}
\label{sec:Optimal_NONrep}

\subsection{Numerical framework for computing the optimal controls}
\label{numerical-nonrep}

The main challenge in computing the  optimal state and control functions $\bigl((\mathbf{S}, \mathbf{I}), \mathbf{u}\bigr)$ is that while the differential equations \eqref{Imm_Sys_NonRep} can be solved once the optimal controls $\mathbf{u}(\cdot)$ are known, an exhaustive search for an optimal control is infeasible as there are an  uncountably infinite number of control functions. \emph{Pontryagin's Maximum Principle (PMP)} provides an elegant technique for solving this seemingly intractable problem (c.f.~\cite{stengel}). \red{\hide{PMP derives from the classical calculus of variations. It bears a close analogy to the primal-dual non-linear optimization framework and indeed generalizes it in the sense that that it allows optimization in the function (as opposed to the variable) space: \eqref{Imm_cost_NonRep} and \eqref{Imm_Sys_NonRep} are analogous to the objective function and constraints of a primal optimization formulation; \emph{adjoint} functions, to be defined shortly, will play the role of Lagrange multipliers and the \emph{Hamiltonian} $\mathcal{H}$ will play \hide{similar to}that of the objective function of the relaxed optimization in the primal-dual framework. }}Referring to the integrand in \eqref{Imm_cost_NonRep} as $\xi_{non-rep}$ and the RHS of (\ref{Imm_Sys_NonRep}a) and (\ref{Imm_Sys_NonRep}b) as $\nu_{i}$ and $\mu_{i}$, we define the Hamiltonian to be
\begin{equation}
 \begin{split}
 \label{hamiltonian}
 \mathcal{H} = \mathcal{H}(\mathbf{u}) :=   \xi_{non-rep} + \sum_{i=1}^M (\lambda^S_i \hide{\dot{S_i}}\nu_{i}+ \lambda^I_i \hide{\dot{I_i}}\mu_{i}),
 \end{split}
\end{equation}
where the \emph{adjoint} (or \emph{costate}) functions $\lambda^S_i$ and $\lambda^I_i$ are continuous
 functions that  for each $i=1\ldots M$, and at each point of continuity of $\mathbf{u}(\cdot)$, satisfy\hide{\footnote{The convention used throughout this paper is $f_{i}(\mathbf{I}):=\dfrac{\partial f(\mathbf{I})}{\partial I_i}$, and $L_i(\mathbf{R}):=\dfrac{\partial L(\mathbf{R})}{\partial R_i}$.}}
\begin{align}
\dot{\lambda}^S_i= -\frac{\partial \mathcal{H}}{\partial S_i}, \quad
\dot{\lambda}^I_i= -\frac{\partial \mathcal{H}}{\partial I_i}, \label{co-state-non-rep}
\hide{=&-L_i(\mathbf{R})-f_{i}(\mathbf{I})-\sum_{j=1}^M\bar\beta_{ji}R_j^0h_j^2(u_j)\notag
\\&- \sum_{j=1}^M\left((\lambda^I_j-\lambda^S_j)\beta_{ij}S_j\right)
 + \lambda^I_i \sum_{j=1}^M\pi_{ji}\bar\beta_{ji} R_j^0u_j}
\end{align}
along with the final (i.e., transversality) conditions
\begin{align}
\label{final-non-rep}
 \lambda^{S}_i(T)=0,\quad \lambda^{I}_i(T)=0.
\end{align}
Then PMP implies that the optimal control at time $t$ \hide{is given by}satisfies
\hide{\begin{align*}
u_i\in \arg\min_{\hide{u_{i,\min}}0\leq \underline{u}_i\leq 1} \mathcal{H}.
\end{align*}}
\begin{equation}\label{PMP_non-rep}
  \mathbf{u} \in \argmin_{\mathbf{v}} \mathcal{H}(\mathbf{v})
\end{equation}
where the minimization is over the space of admissible controls (i.e., $H(u) = \min_v H(v)$)\hide{and we select any Hamiltonian-minimizing control as optimal}.

In economic terms, the adjoint functions represent a shadow price (or imputed value); they measure the marginal worth of an increment in the state at time $t$ when moving along an optimal trajectory. Intuitively, in these terms, $\lambda^I_i$ ought to be positive as it represents the additional cost {that} the system incurs per unit time with an increase in  the fraction of infective nodes. Furthermore, as an increase in the fraction of the infective nodes has worse long-term implications for the system than an increase in the fraction of the susceptibles, we anticipate that $\lambda^I_i-\lambda^S_i>0$. The following result confirms this intuition. It is of value in its own right but as its utility for our purposes is in the proof of our main theorem of the following section, we will defer its proof (to \Sec\ref{lemmaproof}) to avoid breaking up the flow of the narrative at this point.
\begin{Lemma}\label{lem:positive_NONrep}
The positivity constraints $\lambda^I_i(t)>0$ and $\lambda^I_i(t)-\lambda^S_i(t)>0$ hold for all $i=1,\ldots,M$ and all $t\in[0,T)$.
\end{Lemma}

The abstract maximum principle takes on a very simple form in our context. Using the expression for $\xi_{non-rep}$ from \eqref{Imm_cost_NonRep} and the expressions for $\nu_i$ and $\mu_i$ from \eqref{Imm_Sys_NonRep}, trite manipulations show that the minimization~\eqref{PMP_non-rep} may be expressed in the much simpler (nested) scalar formulation
\begin{align}
\label{optimal_u_NonRep:NONlinear}
\hspace{-0.1in}u_i (t)  &\in  \argmin_{0\leq x\leq 1} \psi_i (x,t)
\hide{\Bigg({R_i^0h_i^1(u_i)+R_i^0h_i^2(u_i)}\sum_{j=1}^M\bar\beta_{ij}(S_j+I_j)-\phi_iu_i
\Bigg)} \qquad (1\leq i\leq M);\\
\ 
\label{def:psi_i}\psi_i(x,t):&=R_i^0(h_i(x)-\hide{R_i^0}\phi_i(t)x);\\
\phi_i:&=\hide{R_i^0}\sum_{j=1}^M\bar\beta_{ij}\lambda^S_jS_j +\hide{R_i^0}\sum_{j=1}^M\bar\beta_{ij}\pi_{ij}\lambda^I_jI_j.
\label{phidef} 
 \end{align}
\hide{As can be seen from \eqref{optimal_u_NonRep:NONlinear},  
the condition for an  optimal $u_i$ at each time instant is:
\begin{align}\label{optimal_nonrep_concave}
u_i=\begin{cases}
\hide{1}1 & \psi_i(1)<0\\
0 & \psi_i(1)>0
    \end{cases}
\end{align}
}
\hide{In~\eqref{optimal_u_NonRep:NONlinear} we again select any minimizing argument as optimal (type) control.}

Equation \eqref{optimal_u_NonRep:NONlinear} allows us to characterise $u_i$ as a function of the state and adjoint functions at each time instant. Plugging into \eqref{Imm_Sys_NonRep} and \eqref{co-state-non-rep},  we obtain a system of (non-linear)
differential equations that involve{s} only the state and adjoint functions (and not the control $\mathbf{u}(\cdot)$), and {where} the initial values of the states \eqref{Imm_Sys_NonRep_initialstate} and the final values of the adjoint functions {\hide{(eq.}\eqref{final-non-rep}} are known. Numerical methods for solving \emph{boundary value} nonlinear differential equation problems may now be used   to solve for the state and adjoint functions corresponding to the optimal control, thus providing the optimal controls using  \eqref{optimal_u_NonRep:NONlinear}.

We conclude this section by proving an important property of $\phi_i(\cdot)$, which we will use in subsequent sections.
\begin{Lemma}\label{Lem:Psi}
For each $i$, $\phi_i(t)$ is a decreasing function of $t$.
\end{Lemma}
\begin{proof}
We examine the derivative of $\phi_i(t)$; we need expressions for  the derivatives of
the adjoint functions towards that end. From \eqref{hamiltonian}, \eqref{co-state-non-rep},
at any $t$ at which $\mathbf{u}$ is continuous, we have:
\small
\begin{align}\label{costate}
\hspace{-0.2in}{\dot{\lambda}^S_i}=&-\dfrac{\partial L(\mathbf{R})}{\partial R_i}\hide{(\mathbf{R})}\hide{- \sum_{j=1}^M \bar\beta_{ji}R_j^0h_j^2(u_j)}
- (\lambda^I_i-\lambda^S_i)\sum_{j=1}^M\beta_{ji}I_j+\lambda^S_i\sum_{j=1}^M\bar\beta_{ji}R_j^0 u_j,\notag\\ 
\dot{\lambda}^I_i=&-\sum_{j=1}^M\left((\lambda^I_j-\lambda^S_j)\beta_{ij}S_j\right)
 + \lambda^I_i \sum_{j=1}^M\pi_{ji}\bar\beta_{ji} R_j^0u_j-\dfrac{\partial L(\mathbf{R})}{\partial R_i}-\dfrac{\partial f(\mathbf{I})}{\partial I_i}.
\end{align}
\normalsize
 Using \eqref{phidef}, \eqref{costate} and some  reassembly of terms,  at any $t$
 at which  $\mathbf{u}$ is continuous,
\begin{multline*}
\dot{\phi_i}(t) = -\sum_{j=1}^M \bar\beta_{ij}\bigg[ S_j\dfrac{\partial L(\mathbf{R})}{\partial R_j} + \pi_{ij}I_j\left(\dfrac{\partial L(\mathbf{R})}{\partial R_j} + \dfrac{\partial f(\mathbf{I})}{\partial I_j}\right)
+ \sum_{k=1}^M (1+\pi_{ij})\lambda^I_j\beta_{kj}S_j I_k
+\sum_{k=1}^M\pi_{ij} I_j S_k\beta_{jk}(\lambda^I_k-\lambda^S_k)\biggr].
\end{multline*}
The assumptions on $L(\cdot)$ and $f(\cdot)$ (together with Theorem \ref{thm:constraints}) show that the first two terms inside the square brackets on the right are always non-negative. Theorem \ref{thm:constraints} and Lemma \ref{lem:positive_NONrep} (together with our assumptions on $\pi_{ij}$, $\beta_{ij}$ and $\bar\beta_{ij}$) show that the penultimate term is  positive for $t>0$ and the final term is non-negative. It follows that $\dot{\phi}_i(t) < 0$ for every $t\in (0,T)$ at which $\mathbf{u}(t)$ is continuous. As $\phi_i(t)$ is a continuous function of time and its derivative is negative except at a finite number of points (where $\mathbf{u}$ may be discontinuous), it follows indeed that, as advertised, $\phi_i(t)$ is a decreasing function of time.
\end{proof}
\hide{Equation \eqref{optimal_u_NonRep:NONlinear} allows us to characterise $u_i$ as a function of the state and adjoint functions at each time instant. Plugging these controls into \eqref{Imm_Sys_NonRep} and \eqref{co-state-non-rep},  we obtain a system of (non-linear)
differential equations that involve{s} only the state and adjoint functions (and not the control $\mathbf{u}(\cdot)$), and {where} the initial values of the states \eqref{Imm_Sys_NonRep_initialstate} and the final values of the adjoint functions {\hide{(eq.}\eqref{final-non-rep}} are known. Numerical methods for solving \emph{boundary value} nonlinear differential equation problems may now be used   to solve for the state and adjoint functions corresponding to the optimal control, which will provide the optimal controls using  \eqref{optimal_u_NonRep:NONlinear}.}

\subsection{Structure of optimal non-replicative patching}
We are now ready to identify the structures of the  optimal controls $(u_1(t),\ldots,u_M(t))$:\hide{ We separately consider the cases of concave and convex $h_i(\cdot)$.}

\begin{Theorem}\label{Thm:NonRep:Concave}
Predicated on the existence of an optimal control, for types $i$ such that $R_i^0>0$: if $h_i(\cdot)$ is concave, then the optimal
control for type $i$  has the following structure:  $u_{i}(t)=1$ for $0<t<t_{i}$, and $u_{i}(t)=0$  $t_{i}<t\leq T$,
where $t_i \in [0, T).$  If $h_i(\cdot)$ is strictly convex then
the optimal control for type $i$, $u_i(t)$ is continuous and has the following structure:  $u_{i}(t)=1$ for $0<t<t^1_{i}$,
$u_{i}(t)=0$ for $t^2_{i} < t\leq T$, and $u_i(t)$ strictly decreases in the interval  $[t^1_i, t^2_{i}]$,
where $0 \leq t_i^1 < t_i^2 \leq T.$
\hide{\end{Theorem}}
\end{Theorem}
Notice that if $R_i^0=0$ in \eqref{Imm_Sys_NonRep}, the control $u_i$ is irrelevant and can take any arbitrary admissible value.
 Intuitively, at the onset of the epidemic, a large fraction of{\hide{the}} nodes are susceptible to the malware {\hide{each of
 which is a potential victim}(``potential victims'')}. Bandwidth and power resources{\hide{hence}} should {hence} be used maximally in the beginning (in all types), rendering as many infective and susceptible nodes robust against the malware as possible. In particular, there is no gain in deferring patching since the efficacy of healing infective nodes is less than that of{\hide{the immunization of the} immunizing} susceptible nodes (recall that $\pi_{ij}\leq1$). \hide{The fact that the process of curbing the patching in this case is abrupt rather than gradual is however less apparent.} While the non-increasing nature of the optimal control is intuitive, what is less apparent is the characteristics of the decrease, which we establish in this theorem. For concave $h_i(\cdot)$, nodes are patched at the maximum possible rate until a time instant when patching stops abruptly, while for strictly convex $h_i(\cdot)$, this decrease is continuous. \hide{that the drop in the optimal patching rate is abrupt, as established in the proof of this theorem.}It is instructive to note that the structure of the optimal action taken by a type only depends on its own patching cost and not on that of its neighbours. This is somewhat counter-intuitive as the controls for one type affect the infection and recovery of other types. The timing of the decrease in each type differs and depends on the location of the initial infection as well as the topology of the network, communication rates, etc.
%
\newline
\begin{proof}
\hide{When all $h_i(\cdot)$ and $h_i^2(\cdot)$} 
For non-linear concave $h_i(\cdot)$,\hide{are concave functions then}
\eqref{optimal_u_NonRep:NONlinear} requires the minimization of the (non-linear concave) difference between a non-linear concave function of a scalar variable $x$ and a linear function of $x$ at all time instants; hence the minimum can only occur at the end-points of the interval over which $x$ can vary. Thus all that needs to be done is to compare the values of $\psi_i(x,t)$ for the  following two candidates: $x = 0$ and $x = 1$. Note that $\psi_i(0,t)=0$ at all time instants and $\psi_i(1,t)$ is a function of time $t$. Let
\begin{align}\label{gamma}
\gamma_i(t):=\psi_i(1,t)=R_i^0h_i(1)-R_i^0\phi_i(t).
\end{align}
Then the optimal $u_i$ satisfies the following condition:
\begin{align}\label{optimal_nonrep_concave}
u_i(t)=\begin{cases}
1 & \hide{\psi_i(1)}\gamma_i(t)<0\\
0 & \hide{\psi_i(1)}\gamma_i(t)>0
    \end{cases}
\end{align}

\hide{Note that following }From the transversality conditions in~\eqref{final-non-rep} and the definition of $\phi_i(t)$ in \eqref{phidef}, for all $i$, it follows that $\phi_i(T)=0$.\hide{From Theorem \ref{thm:constraints}{,} $S_j(T) > 0$ for all $j$.} From the definition of the cost term, $h_i(1)>0$\hide{ for all types $i$ that have dispatchers}, hence, \hide{if we define $\gamma_i(\cdot):=\psi_i(1)(\cdot)$ referring to~\eqref{def:psi_i}, and}since $R_i^0>0$, therefore $\gamma_i(T) > 0$. \hide{$\gamma_i(\cdot)$ is a continuous  function of time. }\hide{We will next show} Thus the structure of the optimal control predicted in the theorem for the strictly concave case will follow from \eqref{optimal_nonrep_concave} if we can show that $\gamma_i(t)$ is an increasing
 function of time $t$, as that implies that it can be zero at most at one point $t_i$, with $\gamma_i(t)<0$ for $t<t_i$ and $\gamma_i(t)>0$ for $t>t_i$. From \eqref{gamma},  $\gamma_i$ will be an increasing function of time if $\phi_i$ is a decreasing function of time, a property which we showed in Lemma \ref{Lem:Psi}.

 If $h_i(\cdot)$ is linear (i.e., $h_i(x)= K_ix$, $K_i>0$, since $h_i(x)>0$ for $x>0$), $\psi_i(x,t)=R_i^0x(K_i-\phi_i(t))$ and from \eqref{optimal_u_NonRep:NONlinear}, the condition for an optimal $u_i$ is:
\begin{align}\label{optimal_linear}
u_i(t)=\begin{cases}
1 & \phi_i(t)>K_i\\
0 & \phi_i(t)<K_i
    \end{cases}
\end{align}
 But from \eqref{final-non-rep},  $\phi_i(T)=0<K_i$ and as by Lemma \ref{Lem:Psi}, $\phi_i(t)$ is decreasing, it follows that $\phi_i(t)$ will be equal to $K_i$ at most at one time instant $t=t_i$, with $\phi_i(t)>0$ for $t<t_i$ and $\phi_i(t)<0$ for $t>t_i$\hide{ leading to the conclusion that the transition will be abrupt, as in the concave case}. This, along with \eqref{optimal_linear}, concludes the proof of the theorem for the concave case.

We now consider the case {where\hide{that}} $h_i(\cdot)$ is strictly convex.\hide{ for all $i$ and strictly convex for
some $i.$} \hide{For simplicity,
we assume that $h_i^2(\cdot)$ is a  zero function for all $i$.}
In this case, the minimization in \eqref{optimal_u_NonRep:NONlinear} may also be attained
 at an  interior point of $[0, 1]$  (besides $0$ and $1$) at which the partial
derivative of the right hand side with respect to
$x$ is zero. Hence,
\begin{align}\label{optimal_u_NonRep:NONlinear:convex}
u_i(t)=\begin{cases}
    1 &1< \eta(t)\\
\eta(t) & 0< \eta(t)\leq 1\\
0 & \eta(t)\leq 0.
    \end{cases}
\end{align}
where $\eta(t)$ is such that $
\dfrac{dh_i(x)}{dx}\bigg|_{(x=\eta(t))}=\phi_i(t)$.

Note that $\phi_i(t)$ is a continuous function\hide{ of time} due to the continuity of the states and adjoint functions. We showed that it is also a decreasing function of time (Lemma \ref{Lem:Psi}).\hide{and is equal to the composite of $h_i'(.)$ (an increasing function, due to the convexity of $h_i(\cdot)$) and $\eta(t)$).} Since $h_i(\cdot)$ is double differentiable, its first derivative is continuous, and since it is strictly convex, its derivative is a strictly increasing function of its argument. Therefore, $\eta(t)$ must be a continuous and decreasing function of time, as per the predicted structure. \hide{To conclude the proof of the theorem it only remains to prove Lemma \ref{lem:positive_NONrep}.}
\end{proof}
\subsection{Proof of Lemma \ref{lem:positive_NONrep}}\label{lemmaproof}
\begin{proof}
\hide{\Z{not readable AT ALL.}}\hide{For all $i$, at points of continuity of the control $u_i$, we have:
\hide{\footnote{The convention used throughout this paper is $f_{i}(\mathbf{I}):=\dfrac{\partial f(\mathbf{I})}{\partial I_i}$, and $L_i(\mathbf{R}):=\dfrac{\partial L(\mathbf{R})}{\partial R_i}$.}}
\begin{align}\label{costate}
{\dot{\lambda}^S_i}=&-\dfrac{\partial L(\mathbf{R})}{\partial R_i}\hide{(\mathbf{R})}\hide{- \sum_{j=1}^M \bar\beta_{ji}R_j^0h_j^2(u_j)}
- (\lambda^I_i-\lambda^S_i)\sum_{j=1}^M\beta_{ji}I_j+\lambda^S_i\sum_{j=1}^M\bar\beta_{ji}R_j^0 u_j\notag 
\end{align}
\begin{align}
\dot{\lambda}^I_i=&-\dfrac{\partial L(\mathbf{R})}{\partial R_i}\hide{(\mathbf{R})}-\dfrac{\partial f(\mathbf{I})}{\partial I_i}\hide{(\mathbf{I})}\hide{-\sum_{j=1}^M\bar\beta_{ji}R_j^0h_j^2(u_j)}\notag
- \sum_{j=1}^M\left((\lambda^I_j-\lambda^S_j)\beta_{ij}S_j\right)
\\& + \lambda^I_i \sum_{j=1}^M\pi_{ji}\bar\beta_{ji} R_j^0u_j
\end{align}}\hide{We evaluate the adjoint functions and their derivatives (}From \eqref{costate} and \eqref{final-non-rep}, at time $T$ we have:
\begin{align*}
\lambda^I_i|_{t=T}&=(\lambda^I_i-\lambda^S_i)|_{t=T}=0, \\
\lim_{t\uparrow T}\dot{\lambda}^I_i&=-\dfrac{\partial L(\mathbf{R})}{\partial R_i}(T)\hide{(\mathbf{R}(T))}-\dfrac{\partial f(\mathbf{I})}{\partial I_i}(T)\hide{(\mathbf{I}(T))\\&-\sum_{j=1}^M\bar\beta_{ji}R_j^0h_j^2(u_j(T))}<0\hide{, \ \ \text{and}}\\ \lim_{t\uparrow T}(\dot{\lambda}^I_i-\dot{\lambda}^S_i)&=-\dfrac{\partial f(\mathbf{I})}{\partial I_i}\hide{(\mathbf{I}}(T)<0
\end{align*}
Hence,
$\exists~\epsilon>0$ s.t. $\lambda^I_i>0$ and $(\lambda^I_i-\lambda^S_i)>0$ over $(T-\epsilon,T)$\hide{ (due to the continuity of the adjoint functions and a contradiction argument based on Property \ref{property1})}.
\hide{
We now state a property that holds for any real-valued function, and which we will use in the rest of the proof.
\begin{Property}\label{property1}
Let $\psi(t)$ be a continuous and piecewise differentiable function of t. If $\psi(t_0) = L$ and $\psi(t) < L$ ($\psi(t) > L$) for all $t\in (t_0, t_1]$. Then
$\dot{\psi}_(t_0^+)\leq0$ (respectively $\dot{\psi}_(t_0^+)\geq0$).\hide{ For $\dot{\psi}(t_0^-)$ and the interval $[t_1,t_0)$, the converse holds.}
\end{Property}
\begin{proof}
In the appendix.
\end{proof}
}

\hide{
\begin{proof}
By contradiction. We assume $\psi(t_0) = L$ and $\dot{\psi}_(t_0^+)>0$. Then exists a $\delta \in (0,t_1- t_0)$ such that  $\dot{\psi}_(t)>0$ for $t \in(t_0, t_0+\delta)$.But $\psi(t_0+\delta)= \psi(t_0)+\int_{t_0}^{t_0+\delta} \! \dot{\psi}_(x) \, \mathrm{d} x\hide{=L+\int_{t_0}^{t_0+\delta} \! \dot{\psi}_(x) \, \mathrm{d} x}>L$, which is a contradiction, and thus the property holds. The proof for the $\psi(t) > L$ is exactly as above, with the signs interchanged.
\end{proof}
}

Now suppose that, going backward in time from $t=T$, (at least) one of the inequalities is
first violated at $t=t^*$ for $i^*$, i.e., for all $i$, $\lambda^I_i(t)>0$ and  $(\lambda^I_i(t)-\lambda^S_i(t))>0$ for all $t>t^*$ and either (A)~$(\lambda^I_{i^*}(t^*)-\lambda^S_{i^*}(t^*))=0$ \hide{for an $i=i^*$}\hide{ and $(\lambda^I_{i}-\lambda^S_{i})\geq0$ for all $i$ and $\lambda^I_i\geq 0$ for all $i$} or
(B)~$\lambda^I_{i^*}(t^*)=0$ for some $i=i^*$\hide{and $\lambda^I_i\geq 0$ for all $i$ and $(\lambda^I_{i}-\lambda^S_{i})\geq0$ for all $i$}. Note that from continuity of the adjoint functions $\lambda^I_{i}(t^*)\geq0$  and $(\lambda^I_{i}(t^*)-\lambda^S_{i}(t^*))\geq0$ for all $i$.
\hide{At such a point,
 }

We investigate case~(A) first. We have:\footnote{ $g(t_0^+):=\lim_{t\downarrow t_0}g(t)$ and $g(t_0^-):=\lim_{t\uparrow t_0}g(t)$.}\textsuperscript{,}\footnote{The RHS of the equation is evaluated at $t=t^*$ due to continuity.} 
\begin{align*}
 (\dot{\lambda}^I_{i^*}-\dot{\lambda}^S_{i^*})(t^{*+})=&
 - \hide{f_i(\mathbf{I})}\dfrac{\partial f(\mathbf{I})}{\partial I_{i^*}} - \sum_{j=1}^M[(\lambda^I_j-\lambda^S_j)\beta_{i^*j}S_j] -\hide{(1- \pi_{i^*})}\lambda^I_{i^*}\sum_{j=1}^M\bar\beta_{ji^*}(1- \pi_{ji^*})R_j^0u_j.
\end{align*}
First of all, $-{\partial f(\mathbf{I})}/{\partial I_{i^*}}\hide{(\mathbf{I})}<0$.
The other two terms are non-positive, due
to the definition of $t^*$ and $\pi_{ij}\leq1$.
Hence, $(\dot{\lambda}^I_{i^*}-\dot{\lambda}^S_{i^*})(t^{*+})<0$, which is in contradiction with
Property \ref{property1}\hide{ and the supposition of case~(A)} of real-valued functions, proved in \cite{khouzani2012maximum}: \hide{stated below and proved in Appendix~\ref{General Properties}} 
\begin{Property}\label{property1}
Let $g(t)$ be a continuous and piecewise differentiable function of t. If $g(t_0) = L$ and $g(t) >L$ ($g(t) < L$) for all $t\in (t_0, t_1]$. Then
$\dot{g}(t_0^+)\geq0$ (respectively $\dot{g}(t_0^+)\leq0$).\hide{ For $\dot{g}(t_0^-)$ and the interval $[t_1,t_0)$, the converse holds.}
\end{Property}

On the other hand, for case~(B) 
we have:\footnotemark[\value{footnote}]
\begin{align*}
 \dot{\lambda}^I_{i^*}(t^{*+})&=
-\hide{L_i(\mathbf{R})}\dfrac{\partial L(\mathbf{R})}{\partial R_{i^*}} - \hide{f_i(\mathbf{I})}\dfrac{\partial f(\mathbf{I})}{\partial I_{i^*}} \hide{-\sum_{j=1}^M\bar\beta_{ji}R_j^0h_2^j(u_j)\\&\quad}-\sum_{j=1}^M[(\lambda^I_j-\lambda^S_j)\beta_{i^*j}S_j], 
\end{align*}
 which is negative since $-{\partial L(\mathbf{R})}/{\partial R_{i^*}}\leq0$, $-{\partial f(\mathbf{I})}/{\partial I_{i^*}}<0$,  and from the definition of $t^*$ (for the third term). This contradicts Property \ref{property1} and the claim follows.
\hide{\end{proof}}
\end{proof}
\hide{If $h_i(\cdot)$ is linear (i.e., $h_i(u_i)= K_iu_i$), $\psi_i(u_i)=R_i^0u_i(K_i-\phi_i)$, and as $\phi_i$ is a decreasing function of time (due to the proof of \ref{Lem:Psi}), $\psi_i(u_i)$ will cross zero at most at one point, leading to the conclusion that the transition will be abrupt, as in the concave case.}

\hide{It is instructive to note that the structure of the optimal action taken by a region is independent of the patching cost of its neighbours, and thus the form of their optimal controls, a somewhat counter-intuitive result.
}
\hide{We now consider the case {where\hide{that}} $h_i(\cdot)$ is strictly convex.\hide{ for all $i$ and strictly convex for
some $i.$} \hide{For simplicity,
we assume that $h_i^2(\cdot)$ is a  zero function for all $i$.}
In this case, the minimization in \eqref{optimal_u_NonRep:NONlinear} may also be attained
 at an  interior point of $(0, 1)$  (besides $0$ and $1$) at which the partial
derivative of the right hand side with respect to
$u_i$ is zero. Hence,
\begin{align}\label{optimal_u_NonRep:NONlinear:convex}
u_i(t)=\begin{cases}
    1 &1< \eta(t)\\
\eta(t) & 0<\eta(t)\leq 1\\
0 & \eta(t)\leq 0.
    \end{cases}
\end{align}
where $\eta(t)$ is such that\footnote{Throughout the paper, functions with acute mark (e.g., $h_i'(x)$) represent their derivative with respect to their \underline{argument} (e.g., $\dfrac{d}{dx}h_i(x)$). \hide{As another example, $h_i''(x)$ represents $\dfrac{d^2}{dx^2}h_i(x)$.}}
\begin{equation}\label{eq:implicit}
h_i'(\eta(t))=\hide{\phi_i(t)}
\hide{R_i^0}
\sum_{j=1}^M\bar\beta_{ij}\lambda^S_jS_j +
\hide{R_i^0}\sum_{j=1}^M\bar\beta_{ij}\pi_{ij}\lambda^I_jI_j
\end{equation}\hide{and $\phi_i$ has been defined in
\eqref{phidef}}.
In this case, the structure of optimal $u_i(t)$ for each $i$ is similar
to the concave case, except that the transition between extreme values is
continuous rather than abrupt.

 We derived the expression for $-\dot{\phi}_i$ and proved its positivity at all points of continuity of the control in\hide{~\eqref{minus_phi_dot}, which is positive according to} lemma~\ref{Lem:Psi}. (Note that in the proof of the lemma,
we did not use convexity or concavity property of the $h_i(\cdot)$ functions.) Thus, $\phi_i$ is a decreasing function of time and is equal to the composite of $h_i'(.)$ (an increasing function, due to the convexity of $h_i(\cdot)$) and $\eta(t)$). Therefore, $\eta(t)$ must be a decreasing function of time, proving this part of the theorem.}

\section{Optimal Replicative Patching}\label{sec:Optimal_rep}
\subsection{Numerical framework for computing the optimal controls}
As in the non-replicative setting, we develop a numerical framework for calculation of the optimal solutions using PMP, and then we establish the structure of the optimal controls.

For every control $\tilde{\bu}$, we define $\tau_i(\bI(0),\bS(0),\bR(0),\tilde{\bu})\in[0,T]$ as follows: If $R_i(0)>0$, and therefore $R_i(t)>0$ for all $t>0$ due to Theorem~\ref{thm:constraints}, we define $\tau_i(\bI(0),\bS(0),\bR(0),\tilde{\bu})$ to be 0. Else, $\tau_i(\bI(0),\bS(0),\bR(0),\tilde{\bu})$ is the maximum $t$ for which $R_i(t)=0$.  It follows from Theorem~\ref{thm:constraints} that $R_i(t) = 0$ for all $t \leq \tau_i(\bI(0),\bS(0),\bR(0),\tilde{\bu})$ {and all $i$ such that $R_i(0) = 0$}, and $R_i(t) > 0$ for all $\tau_i(\bI(0),\bS(0),\bR(0),\tilde{\bu}) < t \leq T$. We begin with the hypothesis that there exists at least one optimal control, say $\tilde{\bu}\in \mathcal{U}^*$, and  construct  a control $\bu$ that chooses $u_i(t):= 0$
for $t \leq \tau_i(\bI(0),\bS(0),\bR(0),\tilde{\bu})$ and $u_i(t):= \tilde{u}_i(t)$ for $t > \tau_i(\bI(0),\bS(0),\bR(0),\tilde{\bu}).$ Clearly, the states $\bS(t), \bI(t), \bR(t)$ corresponding to
$\tilde{\bu}$ also constitute the state functions for $\bu${, as the state equations only differ at $t=0$, a set of measure zero}. Thus, $\bu$ is also an optimal control, and
$\tau_i(\bI(0),\bS(0),\bR(0),\tilde{\bu}) = \tau_i(\bI(0),\bS(0),\bR(0),\bu)$ for each $i$. Henceforth, for notational convenience, we will refer to $\tau_i(\bI(0),\bS(0),\bR(0),\tilde{\bu}), \tau_i(\bI(0),\bS(0),\bR(0),\bu)$ as $\tau_i.$ Note that the definition
of this control completely specifies the values of each $u_i$ in  $[0, T]$.

Referring to the integrand of \eqref{Imm_cost_Rep} as $\xi_{rep}$\hide{as before} and the RHS of equations (\ref{Imm_Sys_Rep}a,b,c) as $\nu_i$, $\mu_i$ and $\rho_i$ the Hamiltonian becomes:
\begin{equation} \label{hamiltonian-rep}
 \ham = \ham(\mathbf{u}) :=  \xi_{rep}
+ \sum_{i=1}^M[(\lambda^S_i \hide{\dot{S_i}}\nu_i + \lambda^I_i \hide{\dot{I_i}}\mu_i + \lambda^R_i \hide{\dot{R_i}}\rho_i),
\end{equation}
where  the \emph{adjoint} functions $\lambda^S_i, \lambda^I_i, \lambda^R_i$ are continuous functions
 that at each point of continuity of $\mathbf{u}(\cdot)$ and for all $i=1\ldots M$, satisfy
\begin{align} \label{co-state-rep1}
 \dot{\lambda}^S_i=-\frac{\partial \Ham}{\partial S_i}, \
\dot{\lambda}^I_i=-\frac{\partial \Ham}{\partial I_i}, \
\dot{\lambda}^R_i=-\frac{\partial \Ham}{\partial R_i},
\end{align}
with the final constraints:
\begin{align} \label{final-rep}
 \lambda^S_i(T)=\lambda^I_i(T)=\lambda^R_i(T)=0.
\end{align}
According to PMP, any optimal controller must satisfy:
\begin{align}\label{psi_general}
 \mathbf{u}\in \argmin_{\mathbf{v}}\ham(\mathbf{v}),
\end{align}
where the minimization is over the set of admissible controls\hide{ and we select any Hamiltonian-minimizing control as optimal}.

Using the expressions for $\xi_{rep}$ from \eqref{Imm_cost_Rep} and the expressions for  $\nu_i$, $\mu_i$ and $\rho_i$ from \eqref{Imm_Sys_Rep}, it can be shown that the vector minimization \eqref{psi_general} can be expressed as a scalar minimization
\begin{align}\label{optimal_u:replicative}
u_i (t) &\in  \argmin_{0\leq x\leq 1} \psi_i(x,t) \qquad(1\leq i\leq M);\\
\label{psidef}
\psi_i(x,t):&=R_i(t)(h_i(x)-\phi_i(t)x);\\
\label{phirep}
\phi_i:=\sum_{j=1}^M&\bar\beta_{ij}(\lambda^S_j-\lambda^R_j)S_j+ \sum_{j=1}^M\pi_{ij}\bar\beta_{ij}(\lambda^I_j-\lambda^R_j)I_j\hide{\Bigg)}.
\end{align}

 Equation \eqref{optimal_u:replicative} characterizes the optimal control $u_i$ as a function of the state and adjoint functions at each instant. Plugging the optimal $u_i$ into the state and adjoint function equations (respectively \eqref{Imm_Sys_Rep} and \eqref{co-state-rep1}) will again leave us with a system of (non-linear)
differential equations  that involves only
the state
and adjoint functions (and not the control $\mathbf{u}(\cdot)$), the
initial values of the states \eqref{Imm_Sys_Rep_init}  and the final values of the adjoint functions \eqref{final-rep}.
Similar to the non-replicative case, the optimal controls may now be
obtained (via \eqref{optimal_u:replicative}) by solving the above system of differential equations.

We conclude this subsection by stating and proving some important properties of
 the adjoint functions
 (Lemma~\ref{lem:positive_replicative} below) and $\phi_i(\cdot)$ (Lemma~\ref{phirep1} subsequently), which we use later.

First, from \eqref{final-rep}, $\psi_i(0,t)=0$, hence \eqref{optimal_u:replicative} results in $\psi_i(u_{i},t)\leq0$.
\hide{\begin{align}\label{psi_condition}
\psi_i(u_{i},t)\leq0.
\end{align}}
Furthermore, from the definition of $\tau_i$, if $t\leq\tau_i$, $(h_i(u_i(t))-\phi_i(t)u_i(t))=0$, and if $t>\tau_i$, $(h_i(u_i(t))-\phi_i(t)u_i(t))=\dfrac{\psi_i(u_{i},t)}{R_i(t)}\leq0$, so for all $t$,
\begin{align}\label{psi_condition}
\alpha_i(u_i, t):=(h_i(u_i(t))-\phi_i(t)u_i(t))\leq0.
\end{align}
\begin{Lemma}\label{lem:positive_replicative}
 For all $t\in[0,T)$ and for all $i$, we have $(\lambda^I_i-\lambda^S_i)>0$ and
$(\lambda^I_i-\lambda^R_i)>0$.
\end{Lemma}

Using our previous intuitive analogy, Lemma \ref{lem:positive_replicative} implies that infective nodes are always worse for the evolution of the system than either susceptible or healed nodes, and thus the marginal price of infectives is greater than that of susceptible and healed nodes at all times before $T$. As before, we defer the proof of this lemma (to \Sec\ref{lambdas}) to avoid breaking up the flow of the narrative. We now state and prove Lemma \ref{phirep1}.

\begin{Lemma}\label{phirep1}
For each $i$, $\phi_i(t)$ is a decreasing function of $t$, and $\dot{\phi}_i(t^+)<0$ and $\dot{\phi}_i(t^-)<0$ for all $t$.
\end{Lemma}

\begin{proof}
$\phi_i(t)$  is continuous everywhere (due to the continuity of the states and adjoint functions) and differentiable whenever $\mathbf{u}(\cdot)$ is continuous.
At any $t$ at which $\mathbf{u}(\cdot)$ is continuous, we have:\hide{\footnote{The RHS of the equation is evaluated at $t=t^+$.}}
\begin{align*}
\dot{\phi}_i(t)=
\sum_{j=1}^M
\bar\beta_{ij}[&(\dot{\lambda}^S_j-\dot{\lambda}^R_j)S_j+
(\lambda^S_j-\lambda^R_j)\dot{S}_j
+\pi_{ij}(\dot{\lambda}^I_j-\dot{\lambda}^R_j)I_j+\pi_{ij}(\lambda^I_j-\lambda^R_j)\dot{I}_j].
\end{align*}

From \eqref{hamiltonian-rep} and the adjoint equations \eqref{co-state-rep1}, at points of continuity of the control, we have:
\small
\begin{align}\label{costaterep}
 \dot{\lambda}^S_i
=&\hide{-\sum_{j=1}^M K_{2j}\bar\beta_{ji}R_ju_j}-(\lambda^I_i-\lambda^S_i)\sum_{j=1}^M\beta_{ji}I_j-(\lambda^R_i-\lambda^S_i)\sum_{j=1}^M\bar\beta_{ji}R_ju_j,\notag\\
\dot{\lambda}^I_i=&-\hide{f_i(\mathbf{I})}\dfrac{\partial f(\mathbf{I})}{\partial I_i}
-\sum_{j=1}^M(\lambda^I_j-\lambda^S_j)\beta_{ij}S_j-(\lambda^R_i-\lambda^I_i)\sum_{j=1}^M\pi_{ji}\bar\beta_{ji}R_ju_j,\notag\\
\dot{\lambda}^R_i
=&\hide{L_i(\mathbf{R})}\dfrac{\partial L(\mathbf{R})}{\partial R_i}\hide{-K_{2i}u_i\sum_{j=1}^M\bar\beta_{ij}(S_j+I_j)}
+u_i\sum_{j=1}^M\bar\beta_{ij}(\lambda^S_j-\lambda^R_j)S_j
+u_i\sum_{j=1}^M\pi_{ij}\bar\beta_{ij}(\lambda^I_j-\lambda^R_j)I_j-h_i(u_i)
=\dfrac{\partial L(\mathbf{R})}{\partial R_i}-\alpha_i(u_i,t).
\end{align}
\normalsize
Therefore, after some regrouping and cancellation of terms, at any $t$, we have
\small
\begin{align*}
-\dot{\phi}_i({t}^+)=\sum_{j=1}^M\bar\beta_{ij}[(1-\pi_{ij})\sum_{k=1}^M(\lambda^I_j-\lambda^R_j)\beta_{kj}I_kS_j
+\pi_{ij} \frac{\partial f(\mathbf{I})}{\partial I_j} I_j+(S_j+\pi_{ij} I_j)(\frac{\partial L(\mathbf{R})}{\partial R_j}-\alpha_i(u_i, t))+\pi_{ij} I_j\sum_{k=1}^M(\lambda^I_k-\lambda^S_k)\beta_{jk}S_k
].
\end{align*}
\normalsize
 Now, since $0\leq \pi_{ij}\leq 1$, the assumptions on $\bar{\beta}_{ij}$, $\beta_{ki}$ and $\beta_{il}$, Theorem \ref{thm:constraints}, and Lemma \ref{lem:positive_replicative} all together imply that the sum of the first and last terms of the RHS will be positive. The second and third terms will be non-negative due to the definitions of $f(\cdot)$ and $L(\cdot)$ and \eqref{psi_condition}. So $\dot{\phi}_i({t}^+)<0$ for all ${t}$. The proof for $\dot{\phi}_i({t}^-)<0$ is exactly as above. In a very similar fashion, it can be proved that $\dot{\phi}_i(t)<0$ at all points of continuity of $\mathbf{u}(\cdot)$, which coupled with the continuity of $\phi_i(t)$ shows that it is a decreasing function of time\hide{, concluding the proof of this lemma}.
\hide{, which along with the continuity of $\phi_i$ (notice the control is constant here) can be used to show that $\gamma_i$ will at most be equal to zero at one point, and henceforth will be an increasing function of time (by a simple contradiction argument on Property \ref{property1}).}
\hide{(bearing in mind that $L_j(R)>0$ for all regions, including at least one neighbour whose $S_j>0$.)}
\hide{The proof then completes following from the same arguments for the proof of theorem~\ref{Thm:NonRep:Concave}.}
\end{proof}
\hide{Let $\varphi_i:=K_{1i}R_i+K_{2i}R_i\sum_{j=1}^M\bar\beta_{ij}(S_j+I_j)+\phi_i$.}

\hide{\begin{align*}
u_i(t)=\begin{cases}
1 & \psi_i(t)<0,\\
0 & \psi_i(t)>0\tag{Rep. Cases}.
    \end{cases}
\end{align*}
}
\hide{Thus, since $\psi_i(0)=0$ (from \ref{final-rep}), $\psi(u_{optimal})\leq0$ (due to \eqref{optimal_u:replicative}). }\hide{ Equation \eqref{optimal_u:replicative} characterizes the optimal control $u_i$ as a function of the state and adjoint functions at each instant. Plugging the optimal $u_i$ into the state and adjoint function equations (respectively \eqref{Imm_Sys_Rep} and \eqref{co-state-rep1}) will again leave us with a system of (non-linear)
differential equations  that involve only
the state
and adjoint functions (and not the control $\mathbf{u}(\cdot)$), the
initial values of the states \eqref{Imm_Sys_Rep_init}  and the final values of the adjoint functions \eqref{final-rep}.
Similar to the non-replicative case, the optimal controls may now be
obtained (via \eqref{optimal_u:replicative}) by solving the above system of differential equations.}

%

\hide{Note that for all $i$, $\phi_i(T)=0$.
Following from the continuity of adjoint functions,
we  have $u_i(t)=0$ over an interval of nonzero length ending at $t=T$ for all $i$.}


\subsection{Structure of optimal replicative dispatch}
\hide{The optimal controls for replicative patching exhibit similar structure to that in the non-replicative
setting - in fact }
\hide{\emph{We now prove that the structures indicated in Theorem \ref{Thm:NonRep:Concave} hold for replicative dispatch.}
\newline\mbox{}}
\begin{Theorem}\label{Thm:Rep:Concave}
If an optimal control exists, for types $i$ such that $R_i(t)>0$ for some $t$: if $h_i(\cdot)$ is concave for type $i$, the optimal
control for type $i$ has the following structure:  $u_{i}(t)=1$ for $0<t<t_{i}$, and $u_{i}(t)=0$  for $t_{i}<t\leq T$,
where $t_i \in [0, T).$  If $h_i(\cdot)$ is strictly convex,
the optimal control for type $i$, $u_i(t)$ is continuous and has the following structure:  $u_{i}(t)=1$ for $0<t<t^1_{i}$,
$u_{i}(t)=0$ for $t^2_{i} < t\leq T$, and $u_i(t)$ strictly decreases in the interval  $[t^1_i, t^2_{i}]$,
where $0 \leq t_i^1 < t_i^2 \leq T.$
\end{Theorem}

Notice that for $i$ such that $R_i(t)=0$ for all $t$, the control $u_i(t)$ is irrelevant and can take any arbitrary value. We first prove the theorem for $t \in [\tau_i,T]$, and then we show that $\tau_i \in \{0,T\}$, completing our proof.
\hide{
\begin{Theorem}\label{thm:rep}
 If for a region $i$, $h_i(\cdot)$ is concave\hide{ and that region has some recovered nodes at some point during $[0,T]$}, then there exists a $t_{i}\in[0,T]$, such that an optimal $u_{i}(t)$ can be expressed as $u_{i}(t)=1$ for $0<t<t_{i}$, and
$u_{i}(t)=0$ for $t_{i}\leq t\leq T$. 
If $h_i(\cdot)$ is strictly convex then the optimal $u_i$ is continuous and there exist $t^1_{i}$ and $t^2_{i}$, $0\leq t^1_{i}, t^2_{i}\leq T$,
such that  $u_{i}(t)=1$ for $0<t<t^1_{i}$, and $u_i(t)$ is\hide{ continuous and
}decreasing for $t^1_{i}\leq t\leq t^2_{i}$, and
$u_{i}(t)=0$ for $t^2_{i}\leq t\leq T$. 
\end{Theorem}
}

{\em Proof:}
First consider an $i$ such that $h_i(\cdot)$ is concave and non-linear.\hide{From a similar argument to the one made in the replicative case, $\psi_i(x,t)$ is the difference between a strictly concave and a linear function of $x$, so it} Note that hence $\psi_i(x,t)$ is a non-linear concave function of $x$. Thus, the minimum can only occur at extremal values of $x$, i.e., $x=0$ and $x=1$. Now $\psi_i(0,t)=0$ at all times $t$, so to obtain the structure of the control, we need to examine $\psi_i(1,t)$ at each $t>\tau_i$.
Let\hide{For simplicity, we define} $\gamma_i(t):=\psi_i(1,t)=R_i(t)(h_i(1)-\phi_i(t))$ be a function of time $t$.\hide{ Note that $\gamma_i(t)$ is a continuous function of its variable, due to the continuity of the adjoint functions.} From \eqref{optimal_u:replicative}, the optimal $u_i$ satisfies:
\begin{align}\label{Rep. Cases}
u_i(t)=\begin{cases}
1 & \gamma_i(t)<0,\\
0 & \gamma_i(t)>0.
    \end{cases}
\end{align}
We now show that $\gamma_i(t)>0$ for an interval $(t_i,T]$ for some $t_i$, and $\gamma_i(t)<0$ for $[\tau_i,t_i)$ if $t_i>\tau_i$. From \eqref{final-rep} and \eqref{phirep}, $\gamma_i(T)=h_i(1)R_i(T)>0$. Since $\gamma_i(t)$ is a continuous function of its variable (due to the continuity of the states and adjoint functions), it will be positive for a non-zero interval leading up to $t=T$.
If $\gamma_i(t)>0$ for all $t \in [\tau_i,T]$, the theorem follows. Otherwise, from continuity, there must exist a $t=t_i>\tau_i$ such that $\gamma_i(t_i)=0$. We show that for $t>t_i$, $\gamma_i(t)>0$, from which it follows that  $\gamma_i(t)<0$ for $t<t_i$ (by a contradiction argument).
 The theorem will then follow from \eqref{Rep. Cases}.

Towards establishing the above, we show that $\dot{\gamma_i}(t^+)>0$ and  $\dot{\gamma_i}(t^-)>0$ for any $t$ such that $\gamma_i(t)=0$. Hence, there will exist an interval $(t_i,t_i+\epsilon)$ over which $\gamma_i(t)>0$. If $t_i+\epsilon\geq T$, then the claim holds, otherwise there exists a $t={t_i}^{'}>t_i$ such that $\gamma_i({t_i}^{'})=0$ and $\gamma_i(t)\neq0$ for ${t_i}<t<{t_i^{'}}$ (from the continuity of $\gamma_i(t)$). So $\dot{\gamma_i}({t_i}^{'-})>0$, which contradicts a property of real-valued functions (proved in \cite{khouzani2012maximum})\hide{stated below and proved in Appendix~\ref{General Properties}}, establishing the claim:
\begin{Property}\label{property2}
If $g(x)$ is a continuous and piecewise differentiable function over $[a,b]$\hide{that has a derivative as $x\downarrow a$ and $x \uparrow b$} such that $g(a)=g(b)$ while $g(x)\neq g(a)$ for all $x$ in $(a,b)$\hide{ (i.e., $a$ and $b$ are two consecutive zeros of the function)}, $\frac{dg}{dx}(a^+)$ and $\frac{dg}{dx}(b^-)$ cannot be positive simultaneously.
\end{Property}

We now show that $\dot{\gamma_i}(t^+)>0$ and  $\dot{\gamma_i}(t^-)>0$ for any $t>\tau_i$ such that $\gamma_i(t)=0$.
Due to the continuity of $\gamma_i(t)$ and the states, and the finite number of points of discontinuity of the controls, for any $t>\tau_i$ we have:
\begin{align}\label{gamma3}
\dot{\gamma_i}(t^+)=(\dot{R}_i(t^+)\frac{\gamma_i(t)}{R_i(t)}-R_i(t)\dot{\phi}_i(t^+))\\
\label{gamma4}
\dot{\gamma_i}(t^-)=(\dot{R}_i(t^-)\frac{\gamma_i(t)}{R_i(t)}-R_i(t)\dot{\phi}_i(t^-)).
\end{align}
If $\gamma_i(t)=0$, then $\dot{\gamma_i}({t}^+)= - R_i(t)\dot{\phi}_i({t}^+)$ and $\dot{\gamma_i}({t}^-)= -R_i(t)\dot{\phi}_i({t}^-)$, which are both positive from Lemma \ref{phirep1} and Theorem \ref{thm:constraints}, and thus the theorem follows.
\hide{If we show that $\dot{\phi}_i(t^+)<0$ and $\dot{\phi}_i(t^-)<0$ for all $t$ (as we do in lemma \ref{phirep1}),\hide{ then $\dot{\gamma}_ i(t^+)>0$ and $\dot{\gamma}_ i(t^-)>0$ for all $t$ which satisfy $\gamma_ i(t)\geq0$, especially those which satisfy $\gamma_ i(t)=0$} then the theorem will follow from Theorem \ref{thm:constraints}.}\hide{ by a similar argument to the non-replicative case, $\gamma_i(t)$ will be an increasing function of time from that point onward and will, therefore, remain positive, which along with \eqref{Rep. Cases} leads to the desired structure of the optimal control. \hide{Specifically, a contradiction argument based on Property \ref{property1} shows that  $\gamma_ i(t)$ can be zero at most at one point.\hide{, and for all: (we make the time dependency of $\gamma_i(t)$ implicit)}}}

The proofs for linear and strictly convex $h_i(\cdot)$'s are virtually identical to the corresponding parts of the proof of Theorem~\ref{Thm:NonRep:Concave} and are omitted for brevity; the only difference is that in the linear case
we need to replace  $R_i^0$ with $R_i(t)$. The following lemma, proved in \S\ref{taus}, completes the proof of the theorem.

\begin{Lemma}\label{lem:tau}
For all $0\leq i\leq B,$  $\tau_i \in \{0,T\}$. \qquad\qquad\qquad\QED
\end{Lemma}

\subsection{Proof of Lemma \ref{lem:positive_replicative}}\label{lambdas}
\begin{proof}
First, from \eqref{phirep} and \eqref{final-rep}, we have $\phi_i(T)=0$, which, combined with \eqref{psidef} results in $\psi_i(x,T)=R_i(T) h_i(x)$. Since either $R_i(T)>0$ or $\tau_i>T$, \eqref{optimal_u:replicative} and the definition of $u_i$ result in $u_i(T)=0$, as all other values of $x$ would produce a positive $\psi_i(x,T)$. Therefore, $h_i(u_i(T))=0$\hide{, a result which we will use shortly}.

The rest of the proof has a similar structure to that of Lemma~\ref{lem:positive_NONrep}.
 $(\lambda^I_i-\lambda^S_i)|_{t=T}=0$ and
$\lim_{t\uparrow T}(\dot{\lambda}^I_i-\dot{\lambda}^S_i)=-\hide{f_i(\mathbf{I})}{\partial f(\mathbf{I})}/{\partial I_i}<0$, for all $i$.
Also, for all $i$,  $(\lambda^I_i-\lambda^R_i)|_{t=T}=0$ and $\lim_{t\uparrow T}(\dot{\lambda}^I_i-\dot{\lambda}^R_i)=-\hide{f_i(\mathbf{I})}{\partial f(\mathbf{I})}/{\partial I_i}-\hide{L_i(\mathbf{R})}{\partial L(\mathbf{R})}/{\partial R_i} + h_i(u_i(T))<0$, since $h_i(u_i(T))=0$.

 Hence,
 $\exists~\epsilon>0$ such that $(\lambda^I_i-\lambda^S_i)>0$ and $(\lambda^I_i-\lambda^R_i)>0$ over $(T-\epsilon',T)$\hide{ [again, due to a contradiction argument on Property \ref{property1}]}.

Now suppose that (at least) one of the inequalities is
first\footnote{Going backward in time from $t=T$.} violated at $t=t^*$ for $i^*$, i.e., for all $i$, $(\lambda^I_i(t)-\lambda^S_i(t))>0$ and  $(\lambda^I_i(t)-\lambda^R_i(t))>0$ for all $t>t^*$, and\hide{$(\lambda^I_i(t^*)-\lambda^S_i(t^*))\geq0$, and  $(\lambda^I_i(t^*)-\lambda^R_i(t^*))\geq0$). Now,} either (A)~$(\lambda^I_{i^*}(t^*)-\lambda^S_{i^*}(t^*))=0$, or (B)~$(\lambda^I_{i^*}(t^*)-\lambda^R_{i^*}(t^*))=0$ for some $i^*$. Note that from continuity of the adjoint functions, $(\lambda^I_i(t^*)-\lambda^S_i(t^*))\geq0$, and $(\lambda^I_i(t^*)-\lambda^R_i(t^*))\geq0$ for all $i$.

Case~(A):
Here, we have\hide{Let us investigate
$(\dot{\lambda}^I_{i^*}-\dot{\lambda}^S_{i^*})(t^{*+})$}:\footnote{The RHS of the equation is evaluated at $t=t^*$ due to continuity.}
\begin{align*}
 (\dot{\lambda}^I_{i^*}-\dot{\lambda}^S_{i^*})(t^{*+})=&
-\hide{f_i(\mathbf{I})}\dfrac{\partial f(\mathbf{I})}{\partial I_{i^*}}
-\sum_{j=1}^M(\lambda^I_j-\lambda^S_j)\beta_{i^*j}S_j
-\hide{(1-\pi_{i^*})}(\lambda^I_{i^*}-\lambda^R_{i^*})
\sum_{j=1}^M\bar\beta_{ji^*}(1-\pi_{ji^*})R_ju_j.
\end{align*}
First of all, $-\hide{f_i(\mathbf{I})}{\partial f(\mathbf{I})}/{\partial I_{i^*}}<0$.
Also, the second and third terms \hide{$ - \sum_{j=1}^M[(\lambda^I_j-\lambda^S_j)\beta_{i^*j}S_j]$
and $-(\lambda^I_{i^*}-\lambda^R_{i^*})
\sum_{j=1}^M\bar\beta_{ji^*} (1-\pi_{ji^*})R_ju_j$} are non-positive, according to the definition of $t^*$.
Hence,
$(\dot{\lambda}^I_{i^*}-\dot{\lambda}^S_{i^*})(t^{*+})<0$, which contradicts Property \ref{property1}, therefore case~(A) does not arise.

Case~(B): In this case, we have:\footnotemark[\value{footnote}]

\begin{align*}
 (\dot{\lambda}^I_{i^*}-\dot{\lambda}^R_{i^*})(t^{*+})=
 &-\hide{f_i(\mathbf{I})}\dfrac{\partial f(\mathbf{I})}{\partial I_{i^*}}-\hide{L_i(\mathbf{R})}\dfrac{\partial L(\mathbf{R})}{\partial R_{i^*}}
-\sum_{j=1}^M(\lambda^I_j-\lambda^S_j)\beta_{i^*j}S_j
 +\alpha_{i^*}(u_{i^*},t).
\end{align*}

We have $-{\partial f(\mathbf{I})}/{\partial I_{i^*}}<0$ and $-{\partial L(\mathbf{R})}/{\partial R_{i^*}}\leq 0$. The term $-(\lambda^I_{i^*}-\lambda^S_{i^*})\sum_{j=1}^M\beta_{ji^*}S_j$ is
non-positive, according to the definition of $t^*$, and
$\alpha_{i^*}$ will be non-negative due to \eqref{psi_condition}.
This shows 
$(\dot{\lambda}^I_{i^*}-\dot{\lambda}^R_{i^*})(t^{*+})<0$, contradicting\hide{ion with the assumption $({\lambda}^I_{i^*}-{\lambda}^R_{i^*})<0$ for  $t<t^*$ (via a contradiction argument on } Property \ref{property1}, and
so case~(B) does not arise either, completing the proof.
 \end{proof}

\subsection{Proof of Lemma \ref{lem:tau}}\label{taus}
\begin{proof}
We start by creating another  control  $\bar{\bu}$ from $\bu$ such that for every $i$, for every $t\leq \tau_i$, $\bar{u}_i(t):=1$, and for every $t>\tau_i$, $\bar{u}_i(t):={u}_i(t)$. We prove by contradiction that  $\tau_i(\bI(0), \bS(0), \bR(0), \bar{\bu})\in \{0, T\}$ for each $i$. Since $\bar{u}_i\not\equiv u_i$ only in $[0, \tau_i]$ and $R_i(t) = 0$ for {$t \in (0, \tau_i]$} when $\bu$ is used, {the state equations can only differ at a solitary point $t=0$, and therefore} both controls result in the same state evolutions. Thus, for each $i,$ $\tau_i(\bI(0), \bS(0), \bR(0), \bar{\bu})
= \tau_i(\bI(0), \bS(0), \bR(0),\bu)$, and $\tau_i(\bI(0), \bS(0), \bar{\bu})$ may be denoted as $\tau_i$ as well. The lemma therefore follows.

{For the contradiction argument, assume} that the control is $\bar{\bu}$ and that $\tau_i\in(0,T)$ for some $i$. Our proof relies on the fact that  if $\bar{u}_i(t') = 0$ at some $t' \in (0, T)$, then
 $\bar{u}_i(t) = 0$ for $t > t'$, which follows from the definition of $\bar{\bu}$ and prior results in the proof of Theorem~\ref{Thm:Rep:Concave}. 

For $t\in[0,\tau_i]$, since $R_i(t) = 0$ in this interval,  (\ref{Imm_Sys_Rep}c) becomes $\dot{R}_i=\sum_{j=1,\\j\neq i}^M \bar\beta_{ji}(S_i+\pi_{ji}I_i) R_j\bar{u}_j=0$ in this interval. Since all terms in $\sum_{j=1,\\j\neq i}^M \bar\beta_{ji}(S_i+\pi_{ji}I_i) R_j\bar{u}_j$ are non-negative, for each $j\neq i$ we must either have (i) $\bar\beta_{ji}(S_i(t)+\pi_{ji}I_i(t)) = 0$ for some $t \in [0, \tau_i]$, or (ii) $R_j(t)\bar{u}_j(t)= 0$ for all $t \in[0, \tau_i]$.

(i) Here, either $\bar\beta_{ji}=0$; or $(S_i(t)+\pi_{ji}I_i(t))=0$, and hence due to Theorem~\ref{thm:constraints}, $S_i(t)=0$ and $\pi_{ji}I_i(t)=0$. In the latter case, from Theorem~\ref{thm:constraints},  $S_{i}(0)=0$ and $\pi_{ji}=0$. and therefore for all $t>0$, we will have $\bar\beta_{ji}(S_i+\pi_{ji}I_i) R_j\bar{u}_j=0$.

(ii) We can assume $\bar\beta_{ji}(S_i(t)+\pi_{ji}I_i(t)) > 0$ for all $t \in (0, \tau_i]$. For such $j$, if $\tau_j<\tau_i$, $R_j(t)>0$ for $t \in (\tau_j,\tau_i]$, therefore $\bar{u}_j(t)=0$ for such $t$. Due to the structure results obtained for the interval $[\tau_j,T]$ in Theorem~\ref{Thm:Rep:Concave}, $\bar{u}_j(t)=0$ for all $t>\tau_j$, and therefore $\bar\beta_{ji}(S_i+\pi_{ji}I_i) R_j\bar{u}_j=0$ for all $t>0$. 

Now, since $M<\infty$, the set $W=\{\tau_k:\tau_k>0, k=1,\ldots,M\}$ must have a minimum $\omega_0<T$. Let $L(\omega_0)=\{k\in\{1,\ldots,M\}: \tau_k = \omega_0\}$. Let the second smallest element in $W$ be $\omega_1$. Using the above argument, the values of $R_k(t)$ for $t \in [\omega_0, \omega_1]$ for all $k \in L(\omega_0)$ would affect each other, but not $R_i$'s for $i$ such that $\tau_i>0$, $i \not\in L(\omega_0)$. Furthermore, in this interval for $k \in L(\omega_0)$ we have $\dot{R}_k=\sum_{g \in L(\omega_0)} \bar\beta_{gk}(S_g+\pi_{gk}I_i) R_g\bar{u}_g$, with $R_k(\omega_0)=0$. We see that for all $k\in L(\omega_0)$, replacing $R_k(t)=0$ in the RHS of equation~(\ref{Imm_Sys_Rep}) gives us $\dot{R}_k(t)=0$, a compatible LHS, while not compromising the existence of solutions for all other states. 
An application of Theorem~\ref{thm:constraints} for $t \in [\omega_0, \omega_1]$ and $\bar{u}$ shows that this is the unique solution of the system of differential equations~\eqref{Imm_Sys_Rep}. This contradicts the definition of $\tau_k$, completing the proof of the lemma.
\end{proof}

\section{An Alternative Cost Functional}\label{sec:Generalizations}
Recall that in our objective function, the cost of non-replicative patching was defined as $\sum_{i=1}^MR_i^0h_i(u_i)$ (respectively $\sum_{i=1}^MR_ih_i(u_i)$ for the replicative case), which corresponds to a scenario in which the dispatchers are charged for every instant they are immunizing/healing (distributing the patch), \hide{even if they are not actively engaged in immunizing/healing a particular node at that moment}irrespective of the number of nodes they are delivering patches to. This represents a \emph{broadcast} cost model where each transmission can reach all nodes of the neighbouring types. In an alternative \emph{unicast} scenario, different transmissions may be required to deliver the patches to different nodes. This model is particularly useful if the  dispatchers may {only transmit\hide{only}} to the nodes that have not yet received the patch.\footnote{This can be achieved by keeping a common database of nodes that have successfully received the patch, or {by implementing} a turn-taking algorithm preventing double targeting.
\hide{This choice of policy can remove some unnecessary transmissions and save on the patching overhead, but it should be immediately clear that its implementation involves some extra effort.}
Note that we naturally assume that the network does not know with a priori certainty which nodes are infective, and hence it cannot differentiate between susceptibles and infectives. Consequently, even when $\pi_{ij}=0$, i.e., the system manager knows the patch cannot remove the infection and can only immunize the susceptible, still the best it may be able to do is to forward the message to any node that has not yet received it.}
Hence, the cost of patching in this case can be represented by: $\sum_{i=1}^M \sum_{j=1}^M R_i^0\bar\beta_{ij}(S_j+I_j)p(u_i)$ (for the replicative case: $\sum_{i=1}^M \sum_{j=1}^M R_i\bar\beta_{ij}(S_j+I_j)p(u_i)$), where $p(.)$ is an increasing function. \hide{This cost emphasizes the cost of reception of the patches}More generally, the patching cost can be  represented as a sum of the previously seen cost (\Sec\ref{sec:objective}) and this term\hide{(representing a \emph{broadcast} and \emph{unicast} cost respectively)}.

For non-replicative patching, if all $h_i(\cdot)$ and $p(\cdot)$ are concave, then Theorem \ref{Thm:NonRep:Concave} will hold if for all pairs $(i,j)$, $\pi_{ij}=\pi_j$ (i.e., healing efficacy only depends on the type of an infected node, not that of the immunizer). The analysis will change in the following ways: A term of
$R_i^0 p(u_i)\sum_{j=1}^M \bar\beta_{ij}(S_j+I_j)$
is added to (\ref{def:psi_i}), and subsequently to \eqref{gamma} (with $u_i=1$ in the latter case). Also, (\ref{costate}) is modified by the subtraction of $\sum_{j=1}^M\bar\beta_{ji}R_j^0p(u_j)$ from the RHS of both equations. This leaves $\dot{\lambda}_i^I-\dot{\lambda}_i^S$ untouched, while subtracting a positive amount from $\dot{\lambda}_i^I$,  meaning that Lemma \ref{lem:positive_NONrep} still holds. As $\phi_i(t)$ was untouched, this means that Lemma \ref{Lem:Psi} will also hold. Thus the  RHS of $\dot{\gamma}_i $ is only modified by the subtraction of
%
$\sum_{j,k=1}^M\bar\beta_{ij}(S_j+\pi_{j} I_j)\bar\beta_{kj}R_k^0\left(p(u_k)-u_kp(1)\right)$
which is a positive term, as for any continuous, increasing, concave function $p(\cdot)$ such that $p(0)=0$,
we have $ap(b)\geq bp(a)$ if $a\geq b\geq 0$, since $\frac{p(x)}{x}$ is increasing. This yields: $\left(p(u_k)-u_kp(1)\geq0\right)$.
Therefore the conclusion of Theorem \ref{Thm:NonRep:Concave} holds. Similarly, it may be shown that  Theorem \ref{Thm:NonRep:Concave} also holds for strictly convex $h_i(\cdot)$ provided $p(\cdot)$ is linear. 
 
For the replicative case, if $p(\cdot)$ is linear ($p(x)=Cx$) and again $\pi_{ij}=\pi_j$ for all $(i,j)$, Theorem \ref{Thm:Rep:Concave} will hold. The modifications of the integrand and $\psi_i$ are as above. The adjoint equations (\ref{costaterep}) are modified by the subtraction of $\sum_{j=1}^M C\bar\beta_{ji}R_ju_j$ from $\dot{\lambda}_i^I$ and $\dot{\lambda}_i^I$, and the subtraction of $Cu_i\sum_{j=1}^M\bar\beta_{ij}(S_j+I_j)$ from $\dot{\lambda}_i^R$. Due to the simultaneous change in $\psi_i$, however, we still have $\dot{\lambda}_i^R=\partial L(\mathbf{R})/{\partial R_i}-\alpha_i(u_i,t)$. Therefore, Lemma \ref{lem:positive_replicative} still holds, as $\dot{\lambda}_i^I-\dot{\lambda}_i^S$ is unchanged, and a positive amount is subtracted from $\dot{\lambda}_i^I-\dot{\lambda}_i^R$. We absorb $\sum_{j=1}^M C\bar\beta_{ij}(S_j+I_j)$ into $\phi_i(t)$, where all the $p(\cdot)$ terms in $\dot{\phi}_i$ will cancel out, leaving the rest of the analysis, including for Lemmas \ref{phirep1} and \ref{lem:tau}, to be the same. The theorem follows.

\hide{+\\\sum_{i=1}^M K_{2i}R_iu_i\sum_{j=1}^M\bar\beta_{ij}(S_j+I_j)]}

\hide{+K_{2i}R_iu_i\sum_{j=1}^M\bar\beta_{ij}(S_j+I_j)}

\section{Numerical Investigations}\label{sec:Numericals}
In this section, we numerically investigate the optimal control policies for a range of malware and network parameters.%
\footnote{For our calculations, we use a combination of
\emph{C} programming and \emph{PROPT}\textsuperscript\textregistered,
 by \emph{Tomlab Optimization Inc} for \emph{MATLAB}\textsuperscript\textregistered.} Recalling the notion
 of topologies presented in \Sec\ref{subsec:non_rep_model} (in the paragraph before \eqref{Imm_Sys_NonRep}),
we consider three topologies: \emph{linear}, \emph{star} and \emph{complete}, as was illustrated in Fig.~\ref{fig:topologies}. In our simulations, we assume that at $t=0$, only one of the regions (types) is infected, i.e., $I_i^0>0$ only for $i=1$. Also, $R^0_i = 0.2$, $\beta_{ii} = \beta=0.223$ for all $i$.\footnote{This specific value of $\beta$ is chosen to match the average inter-meeting times from the numerical experiment reported in~\cite{hui2005pocket}.} The value of $\beta_{ij}$, $i\neq j$ is equal to $X_{Coef}\cdot\beta$ if link $ij$ is part of the topology graph, and zero otherwise. (Unless otherwise stated, we use $X_{Coef}=0.1$.) It should be noted that $\beta_{ij}* T$ denotes the average number of contacts between nodes of regions $i$ and $j$ within the time period, and thus $\beta$ and $T$ are dependent variables. For simplicity, we use
 equal values for ${\beta}_{ji}, \beta_{ij}, \bar{\beta}_{ij}$, $\bar{\beta}_{ji}$ for all $i,j$ (i.e., ${\beta}_{ji} = \beta_{ij} = \bar{\beta}_{ij} = \bar{\beta}_{ji}$), and set $\pi_{ij}=\pi$ for all $i,j$.
We examine two different aggregate cost structures for non-replicative patching:\footnote{$f_i(\cdot)$, $h_i(\cdot)$, and $p(\cdot)$ are linear and identical for all $i$, and $l_i(\cdot)=0$.} (type-A)-$\int_{0}^{T}\left(K_I\sum_{i=1}^M I_i(t)+K_u \sum_{i=1}^M R_i^0 u_i(t)  \right)\,dt$ and (type-B)- $\int_{0}^T\left(K_I\sum_{i=1}^M I_i(t)+K_u\sum_{i=1}^M R_i^0 u_i(t)(S_i(t)+I_i(t)) \right)\,dt$ (described in \Sec\ref{subsec:Cost} and \Sec\ref{sec:Generalizations} respectively). We select $T=35$, $K_I=1$, $K_u=0.5$ unless stated otherwise. 
For replicative patching, $R_i^0$ in both cost types is replaced with $R_i(t)$.

We first present an example of our optimal policy (\S\ref{subsec:example}), and then we examine its behaviour in the linear and star topologies (\S\ref{subsec:topology}). Subsequently, we show the cost improvements it achieves over heuristics (\S\ref{subsec:comparison}). Finally, we demonstrate the relative benefits of replicative patching (\S\ref{subsec:nonrep}). 

\begin{figure}[htb]
\centering
\includegraphics[width= 0.5 \columnwidth]{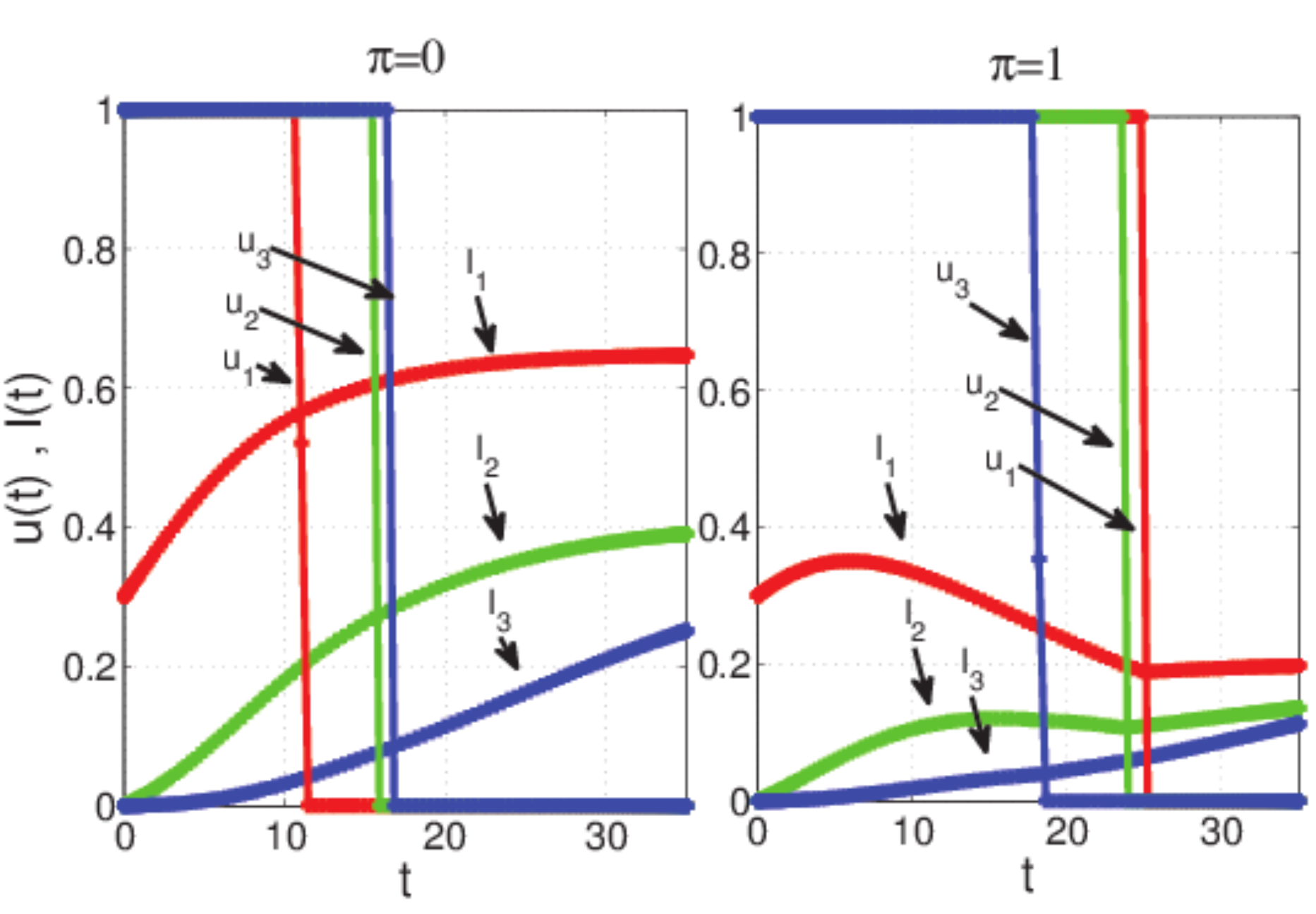}
\caption{Optimal patching policies and corresponding levels of infection in a three region linear topology.
Note how the infection that initially only exists in region 1 spreads in region 1 and then to region 2, and finally to region 3.}\label{fig:example}
\end{figure}

\subsubsection{Numeric Example}\label{subsec:example}
First, with the intention of illustrating our analytical results, in Fig.~\ref{fig:example} we have depicted an example of the optimal dynamic patching policy along with the corresponding evolution of the infection as a function of time for a simple 3-region linear topology where the infection starts in region 1 ($\bI^0=(0.3,0,0)$). The cost model is type-A and patching is non-replicative.
For $\pi=0$ the levels of infection are non-decreasing, whereas for $\pi=1$ they may go down as well as up {(due to healing)}.

\subsubsection{Effects of Topology}\label{subsec:topology}
\hide{We now investigate the effect of topology on the optimal patching policy.}
We study the \emph{drop-off}  times (the time thresholds at which the bang-bang {optimal} patching halts) in different regions for linear and star topologies.

Fig.~\ref{fig:linear_pattern} reveals  two different patterns for $\pi=0$ and $\pi=1$ in a linear topology with $10$ regions with non-replicative patching and type-A cost.
For $\pi=0$, {a\hide{the}} middle region is patched for the longest time, whereas for $\pi=1$, as we move away from the origin of the infection (region 1), the drop-off point decreases. This is because for $\pi=0$, patching can only benefit the network by recovering {\hide{the }}susceptibles. In regions closer to the origin, the fraction of {\hide{the }}susceptibles decreases quickly, making continuation of the patching comparatively less beneficial.
In the middle regions, where there are more salvageable susceptibles, patching should be continued {for} longer. For regions far from the origin, patching can be stopped earlier, as the infection barely reaches them within the time horizon of consideration.
For $\pi=1$, patching is able to recover both susceptible and infective nodes. Hence,  the drop-off times
depend only on the exposure to the infection, which decreases with distance from the origin. As $X_{Coef}$ is increased, the drop-off points when $\pi=1$ get closer together. Intuitively, this is because higher cross-mixing rates have a homogenizing effect, as the levels of susceptible and infective nodes in different region rapidly become comparable.
Also, Fig.~\ref{fig:linear_pattern} reveals that as $X_{Coef}$ increases and more infection reaches farther regions, they are patched for longer, which agrees with our intuition.

\begin{figure}[htb]
\centering
\includegraphics[width= 0.5 \columnwidth]{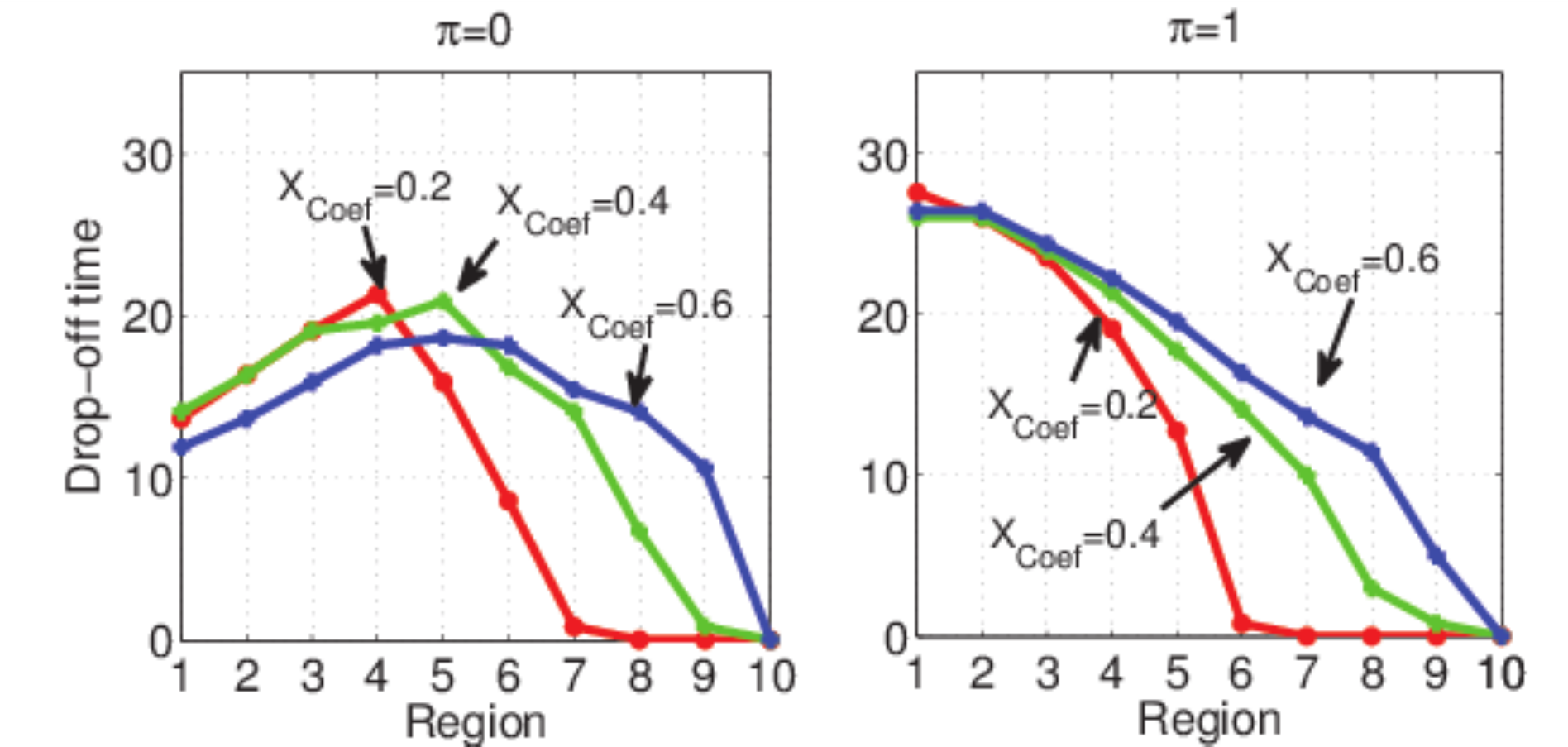}
\caption{Drop-off times in a linear topology for $X_{Coef}=0.2,0.4,0.6$.}\label{fig:linear_pattern}
\end{figure}

We next investigate a star configuration where  the infection  starts from a peripheral region (region 1), cost is type-B, patching is non-replicative, and  $I_1^0=0.6$.
Fig.~\ref{fig:Star_pattern} reveals  the following interesting phenomenon: although the central region
is the only one that is connected to all the regions, for $\pi=0$, it  is patched for shorter lengths of time compared to the peripherals. In retrospect, this is because only susceptible nodes
can be patched and their number at the central region  drops quickly due to its interactions with all the peripheral regions, rendering patching inefficient relatively swiftly.
As expected, this effect is amplified with higher number{s} of peripheral regions.
For $\pi=1$, on the other hand, the central region is patched for the longest time. This is because the infective nodes there  can infect  susceptible nodes in all regions, and hence the patching,  which can now heal the infectives as well, does not stop until it heals  almost all of infective nodes in this region.

\begin{figure}[htb]
\centering
\includegraphics[width= 0.5 \columnwidth]{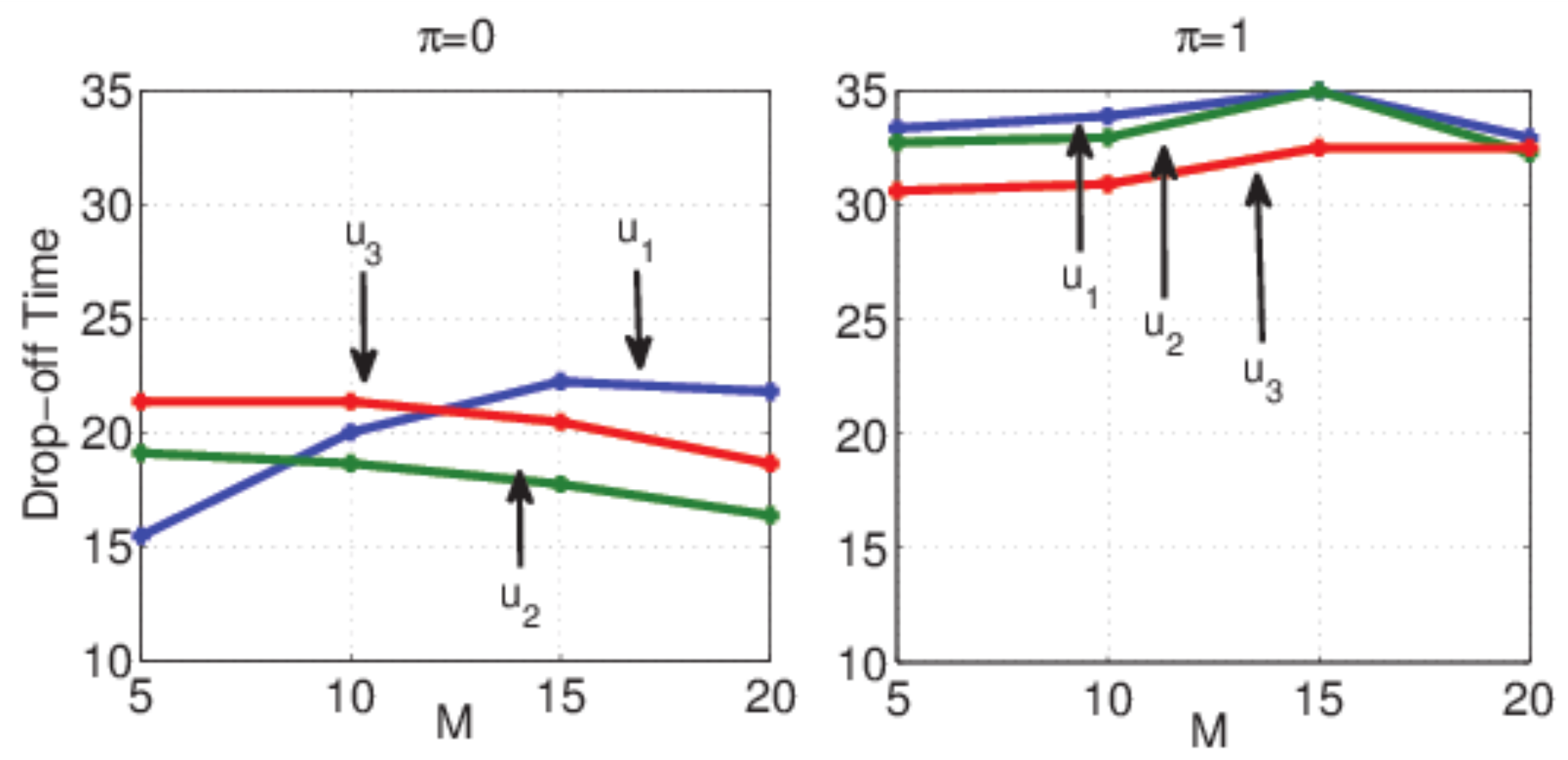}
\caption{Drop-off times in the star topology.}\label{fig:Star_pattern}
\end{figure}

\subsubsection{Cost Comparison}\label{subsec:comparison}
Next, in order to evaluate the efficacy of our dynamic heterogeneous patching policy, we compare our {\hide{inflicted }}aggregate cost
against those of four alternative patching policies. We label our policy{\hide{as}} as \emph{Stratified Dynamic}.
In the simplest alternative policy, all regions use identical  patching intensities that do
not change with time. We then select this fixed and static level of patching so as to  minimize
 the aggregate cost among all possible choices. We refer to this policy as \emph{Static} (St.).
The aggregate cost may be reduced if the static value of the patching is allowed to be distinct for different regions. These values (still fixed over time) are then independently varied and the best combination is selected. We refer to this policy  as \emph{Stratified Static} (S. St.).
The third policy we implement is a \emph{homogeneous} approximation to the heterogeneous network. Specifically, the whole network is approximated by a single region model with an equivalent inter-contact rate. This value is selected such that the \emph{average} pairwise contact rates are equal in both systems.
The optimal control is {\hide{hence }}derived based on this model and applied across all regions to calculate the aggregate cost. We call this policy \emph{Simplified Homogeneous} (S. H.). The simplified homogeneous policy is a special case of Spatially Static (Sp. St.) policies, where a one one-jump bang-bang control is applied to all regions to find the optimum uniform control.

Fig.~\ref{fig:linear_costs} depicts the aggregate costs of all five policies for a linear topology with $M=2,\ldots,5$ {\hide{number of }}regions.
The cost is type-A and patching is replicative, with $\pi=1$. Here, $I_1^0=0.2$, $K_u=0.2$ and the rest of parameters are as before.
As we can clearly observe, our stratified policy achieves the least cost, outperforming the rest.
When the number of regions is small, H. and Sp. St. \hide{is doing}perform better than S. St.,  all of which obviously outperform St.
However, as the number of regions increases and the network becomes more spatially heterogeneous, the homogeneous approximation, and all uniform controls in general, worsen and the S. St. policy quickly overtakes them as the best approximation to the optimal.
For example for $M=5$ regions, our policy outperforms the best static policies by 40\% and the homogeneous  approximation by 100\%, which shows that our results about the structure of the optimal control can result in large cost improvements\hide{in both performance and running time (the latter by restricting the search space)}.  
For $\pi=0$, a similar performance gap is observed. 
\hide{This underscores the significance of considering heterogeneity in the controls. }Specifically, as discussed, optimal drop-off times {for this problem} should vary based on the distance from the originating region, a factor that the Sp. St., H., and St. policies ignore.

\begin{figure}[htb]
\centering
\includegraphics[width= 0.5 \columnwidth]{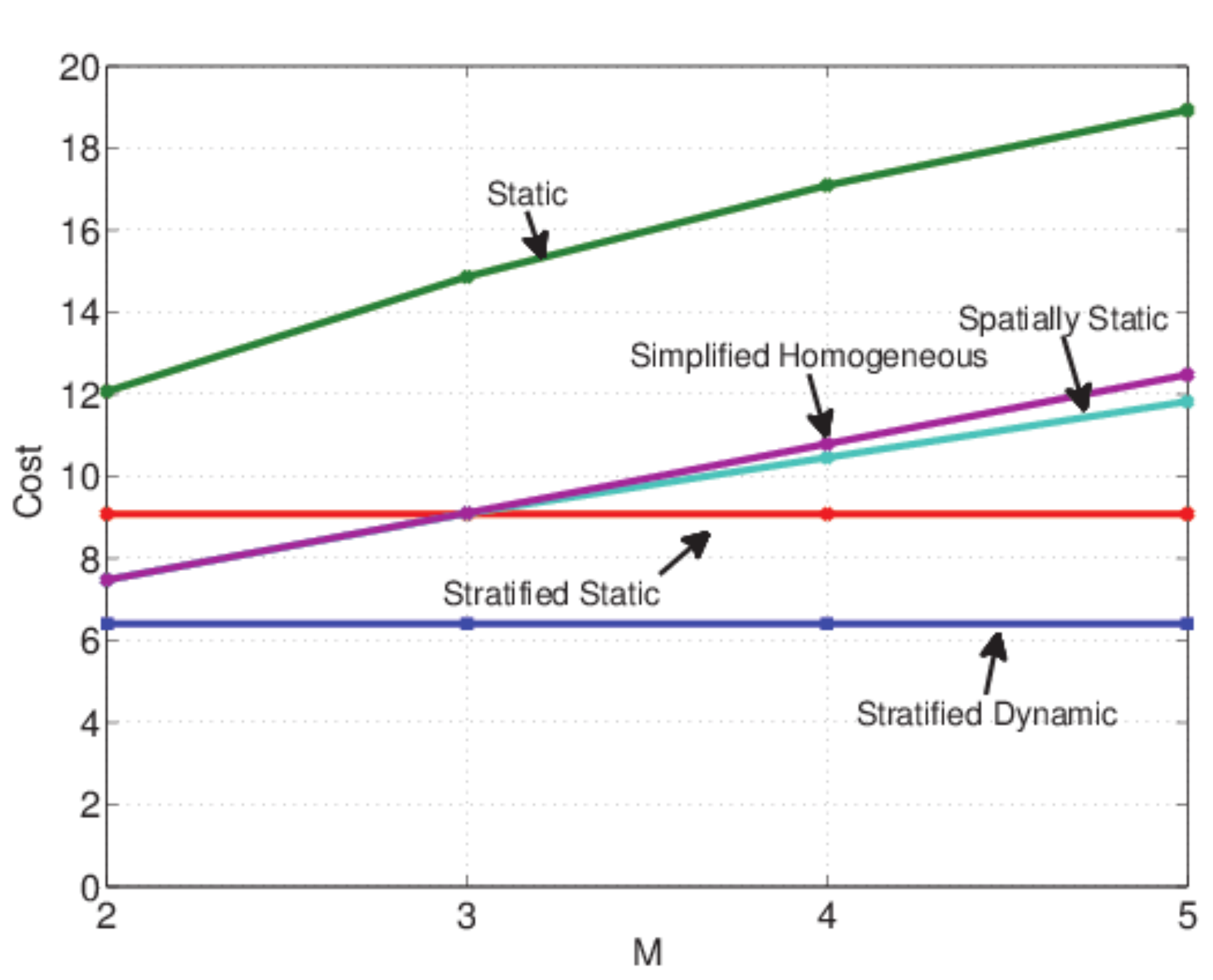}
\caption{Cost of heuristics vs. the optimal policy, linear topology.}\label{fig:linear_costs}
\end{figure}

\subsubsection{Replicative vs. Non-replicative Patching}\label{subsec:nonrep}
As previously stated, any solution to the non-replicative patching problem can be emulated by replicative patching, making this scenario worth investigation, even with the additional security vulnerabilities that the system has to contend with.  In Fig.~\ref{fig:RepNon}, we see the aggregate cost of optimal replicative and non-replicative patching in a complete topology as a function of the size of the network for $M \leq 11$, $\pi=1$, and $K_u=0.2$. Here, even for such modest sizes, replicative patching can be 60\% more efficient than non-replicative patching, a significant improvement. This is especially true for the complete topology and other edge-dense topologies, as in replicative patching, the patch can spread in ways akin to the malware.

\begin{figure}[htb]
\centering
\includegraphics[width= 0.5 \columnwidth]{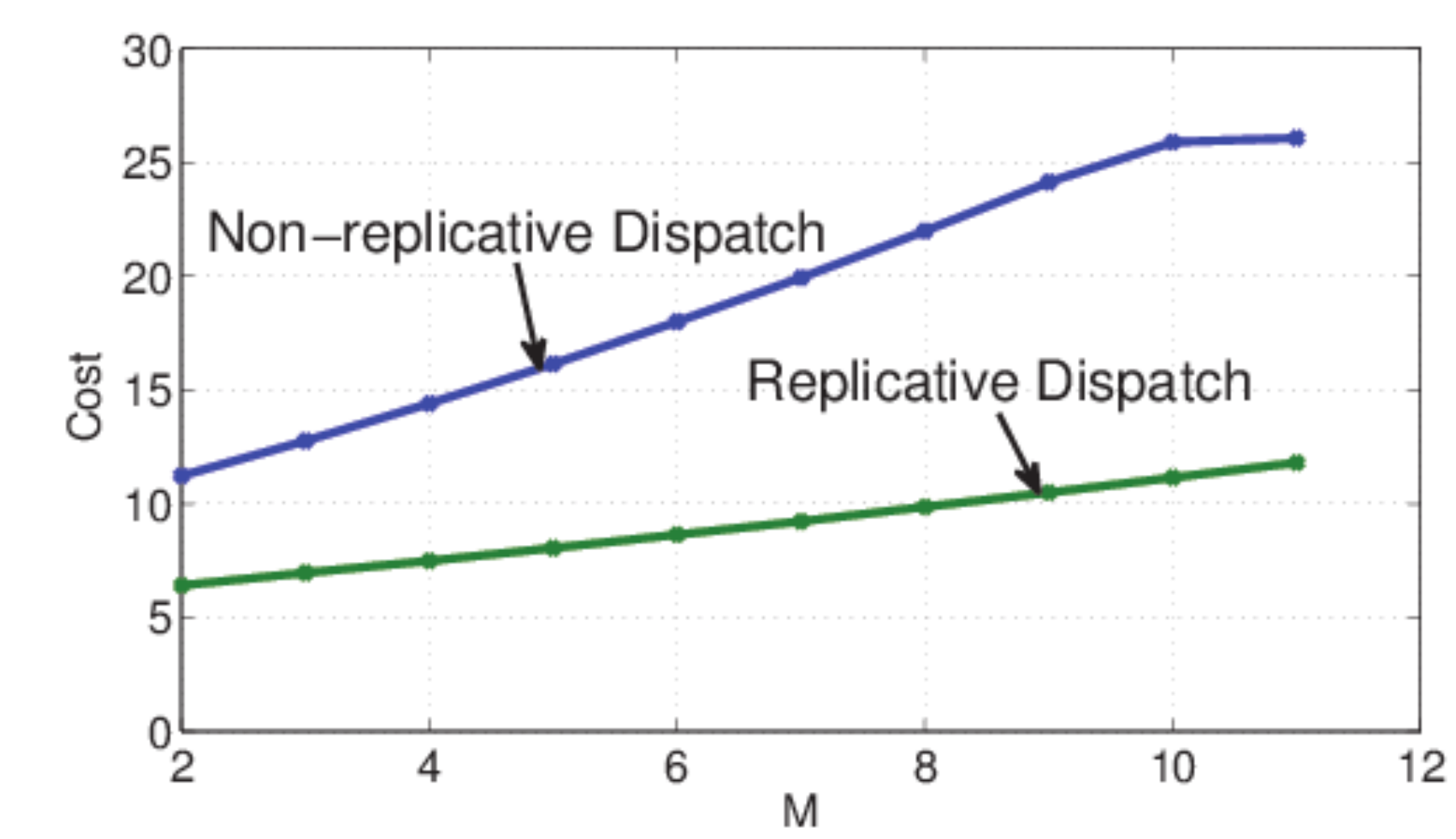}
\caption{Costs of replicative and non-replicative patching, complete topology.}\label{fig:RepNon}
\end{figure}

\section{conclusion and future work}



We considered the problem of disseminating security patches in a large resource-constrained heterogeneous network in the mean-field regime.
Using tools from optimal control theory, we analytically proved that optimal dynamic policies for each type of node follow simple threshold-based structures. \hide{, making them amenable to distributed implementation. }
We numerically demonstrated the advantage of our heterogeneous policies over homogeneous approximations, as well as over static policies.
For future research, we would like to further investigate the effects of heterogeneities in the structure of networks on both defense and attack strategies.

\bibliographystyle{ieeetr}
\bibliography{Reference}

\appendices

\section{Proof of Theorem~\ref{thm:constraints}}
We use the following  general result :
\begin{Lemma}
\label{baselemma}
Suppose the vector-valued function $\boldf = (f_i,1\leq i\leq 3M)$  has component functions given by quadratic forms $
  f_i(t, \bx) = \bx^T Q_i(t) \bx + {p_i}^T \bx \quad (t\in[0,T];\; \bx\in\BS)$,
  where $\BS$  {is the set of $3M$-dimensional vectors} $\bx = (x_1,\dots, x_{3M})$ satisfying $\bx\geq\mathbf{0}$ and $\forall j \in \{1,\ldots,M\};  x_{j} + x_{M+j} +x_{2M+j} = 1$, $Q_i(t)$ is a  matrix whose  components are uniformly, absolutely bounded over $[0,T]$, as are the elements of the vector $p_i$. Then, for an $3M$-dimensional vector-valued function $\bF$, the system of differential equations 
\begin{equation}\label{Vectordf}
\begin{split}
  \dot{\bF} = \boldf(t, \bF) \qquad(0<t\leq T)\\
    \quad\text{subject to initial conditions $\bF(0)\in \BS$}
\end{split}
\end{equation}
 has a unique solution,
 $\bF(t)$, which varies continuously with the initial conditions $\bF_0 \in \BS$  at each $t\in[0,T]$.
 \end{Lemma}
\begin{proof} This lemma is virtually identical to Lemma 1 of \cite{eshghi2013optimal} for $N=3M$, with the difference that here $f_i(t, \bx)$ has an additive ${p_i}^T \bx$ term. Thus, we will only list the changes: As the Euclidean norm $\|Q_i(t)\bx\|$ is uniformly bounded over $(t,\bx)\in[0,T]\times\BS$, there exists $C<\infty$ such that $\sup_{[0,T]\times\BS} \|Q_i(t)\bx\| \leq C$. {Also, $\|{p_i} \| \leq H$ for some  $H<\infty$} as all its elements are bounded. Now, for each $t$, we may write
\begin{align*}
  f_i(t,\bx) - f_i(t,\by)= \bigl(Q_i(t)\bx+{p_i}\bigr)^T (\bx - \by) + (\bx - \by)^T Q_i(t)\by
\end{align*}
Taking absolute values of both sides, we obtain
\begin{align*}
  |f_i(t,\bx) - f_i(t,\by)|\leq \bigl|\bigl(Q_i(t)\bx\bigr)^T (\bx - \by)\bigr|
      + \bigl|(\bx - \by)^T Q_i(t)\by\bigr| + \bigl|{p_i}^T \bigl(\bx - \by\bigr)\bigr|
    \leq  (2C+H) \|\bx-\by\| ~ (t\in[0,T];\; \bx,\by\in\BS),
\end{align*}
by using the triangle and Cauchy-Schwarz inequalities. Hence
\begin{align*}
  \|\boldf(t,\bx) - \boldf(t,\by)\| \leq L\,\|\bx - \by\| \qquad (t\in[0,T];\; \bx,\by\in\BS),
\end{align*}
and so $\boldf(t,\cdot)$ is Lipschitz over $\BS$ where the Lipschitz constant $L = (2C+H) \sqrt{3M}$ may be chosen uniformly for $t\in[0,T]$.

The rest of the proof is exactly as in \cite{eshghi2013optimal}.
\end{proof}

{\em Proof of Theorem ~\ref{thm:constraints}:}{
We write $\bF(0) = \bF_0$,\hide{, in a slightly informal notation, we may thus} and in a slightly informal notation, $\bF = \bF(t) = \bF(t, \bF_0)$ to acknowledge the dependence of $\bF$ on the initial value $\bF_0$. }

We first verify that $ \bS(t) + \bI(t) + \bR(t) = \mathbf{1}$ for all $t$ in both cases.  
By summing the left and right sides of the system of equations~\eqref{Imm_Sys_NonRep} and the $\dot{R}_i$ equation that was left out (respectively the two sides of equations~\eqref{Imm_Sys_Rep}), we see that in both cases for all $i$,
$
 \bigl(\dot{S}_i(t) + \dot{I}_i(t) + \dot{R}_i(t)\bigr) = 0,
$
and, in view of the initial normalization 
{$\bigl(S_i(0) + I_i(0) + R_i(0)\bigr) = 1$, we have 
$ \bigl(S_i(t) + I_i(t) + R_i(t)\bigr) = 1$ for all $t$ and all $i$. } 

We now verify the non-negativity condition.  Let $\bF = (F_1,\dots,F_{3M})$ be the  state vector in ${3M}$ dimensions whose elements are comprised of $(S_i,1\leq i\leq M)$, $(I_i, 1\leq i\leq M)$ and $(R_i, 1\leq i\leq M)$ in some order. The system of equations~\eqref{Imm_Sys_NonRep} can thus be represented as $\dot{\bF} = \boldf(t, \bF)$, where for $t\in[0,T]$ and  $\bx\in\BS$, the vector-valued function $\boldf = (f_i,1\leq i\leq 3M)$ has component functions  $
  f_i(t, \bx) = \bx^T Q_i(t) \bx + {p_i}^T \bx
$
in which (i) $Q_i(t)$ is a  matrix whose non-zero elements are of the form $\pm\beta_{jk}$, (ii) the elements of $p_i(t)$ are of the form $\pm\bar{\beta}_{jk}R^0_ju_j$ and $\pm\bar{\beta}_{jk} \pi_{jk}R^0_ju_j$, whereas~\eqref{Imm_Sys_Rep} can be represented in the same form but with (i) $Q_i(t)$ having elements $\pm\beta_{jk}$, $\pm\bar{\beta}_{jk}u_j$, and $\pm\bar{\beta}_{jk} \pi_{jk}u_j$, and (ii) $p_i={\bf 0}$. Thus, the components of $Q_i(t)$ are uniformly, absolutely bounded over $[0,T]$. Lemma~\ref{baselemma} establishes that the solution $\bF(t, \bF_0)$ to the systems
 \eqref{Imm_Sys_NonRep} and \eqref{Imm_Sys_Rep} is unique and varies continuously with the initial conditions $\bF_0$; it clearly varies continuously with time.   Next, using elementary calculus, we show in the next paragraph that if $\bF_0\in \text{\bf Int }\BS$ (and, in particular, each component of $\bF_0$ is positive), then each component of the solution $\bF(t, \bF_0)$ of \eqref{Imm_Sys_NonRep} and \eqref{Imm_Sys_Rep} is positive at each $t\in[0,T]$.\hide{\footnote{Throughout the paper, we use positive for strictly positive, etc.}}  Since  $\bF(t, \bF_0)$  varies continuously with $\bF_0$, therefore $\bF(t,\bF_0)\geq\mathbf{0}$ for all $t\in[0,T]$, {$\bF_0\in \BS$}, which completes the overall proof.

Accordingly, let the $S_i$, $I_i$, and $R_i$ component of $\bF_0$ be positive. Since the solution $\bF(t, \bF_0)$  varies continuously with time,
there exists a time, say $t' > 0$, such that each component of  $\bF(t, \bF_0)$ is positive
in the interval $[0, t')$. The result follows trivially if $t' \geq T$. Suppose now that there exists $t''<T$
such that each component of  $\bF(t, \bF_0)$ is positive
in the interval $[0, t'')$, and at least one such component is $0$ at $t''$.

We first examine the non-replicative case. We show that such components can not be $S_i$ for any $i$ and subsequently
rule out $I_i$ and $R_i$ for all $i$.  Note that $u_j(t), I_j(t), S_j(t)$ are bounded in $[0, t'']$ (recall $\left(S_j(t)+I_j(t)+R_j(t)\right) = 1 , S_j(t)
\geq 0, I_j(t) \geq 0, R_j(t) \geq 0$ for all $j\in\{1,\ldots,M\}, t \in [0, t'']$). From (\ref{Imm_Sys_NonRep}a)
$S_i(t'') = S_i(0) e^{-\int_{0}^{t''}\sum_{j=1}^{M}(\beta_{ji}I_j(t) +\bar\beta_{ji} R_j^0u_j(t)) \, dt}$.
Since all $u_j(t), I_j(t)$ are bounded in $[0, t'']$, $S_i(0)>0, R_j^0\geq0$, and $\beta_{ji},\bar\beta_{ji} \geq 0$, therefore $S_i(t'') > 0$. Since $S_i(t) > 0$, $I_i(t) \geq 0$ for all $i, t \in [0, t'']$, and $\beta_{ji}\geq 0$, from (\ref{Imm_Sys_NonRep}b), $\dot{I}_i \geq -I_i\sum_{j=1}^M\pi_{ji} \bar\beta_{ji}R_j^{0}u_j$  for all $i$ in the interval $[0,t'']$. Thus, $I_i(t'') \geq I_i(0) e^{-\int_{0}^{t''} \sum_{j=1}^M\pi_{ji} (\bar\beta_{ji}R_j^{0}u_j(t)) \, dt}.$ Since all $u_j(t), I_j(t), S_j(t)$ are bounded in $[0, t'']$, and $I_i(0)>0$, $\bar\beta_{ji}, \pi_{ji}  \geq 0$,  it follows that $I_i(t'') > 0$ for all $i \geq 0$. Finally, $R_i(t'')>0$ because $R_i(0)>0$ and $\dot{R}_i(t)\geq0$ from the above, so $R_i(t)\geq R_i^0$, and $S_i(t)+I_i(t)\leq 1-R_i^0$ for all $t$ and $i$. This contradicts the definition of $t''$ and in turn implies that $\bF(t,\bF
_0)>0$ for all $t \in [0,T]$, $\bF_0 \in \text{\bf Int } \BS$.

The proof for the replicative case is similar, with the difference that $R_i^0$ is replaced with $R_i$, which is itself bounded.

Since the control and the unique state solution $\bS(t)$, $\bI(t)$ are non-negative,~(\ref{Imm_Sys_NonRep}a, \ref{Imm_Sys_Rep}a)
imply that  $\bS(t)$ is a non-increasing function of time. Thus, $S_j(t) = 0$ if $S_j(0) = 0$ for any $j$.  Using the  argument in the above paragraph and starting from a $t' \in [0, T)$ where $S_j(t') > 0$, $I_j(t') > 0$, or $R_j(t') > 0$,
it may be shown respectively that $S_j(t) > 0$, $I_j(t) > 0$, and $R_j(t) > 0$ for all $t > t'$. All that remains to show now is:
\begin{Lemma}\label{lem:I=0}
There exists $\epsilon > 0$ such that $\mathbf{I}(t) >0$ for $t\in(0,\epsilon)$.
\end{Lemma}
Let $d(i,j)$ be the \emph{distance} from type $j$ to type $i$ (i.e., \hide{one more than the minimum number of neighbours that have to be traversed to get from $j$ to $i$} for all $i$, $d(i,i)=0$ and for all pairs $(i,j)$, $d(i,j)=1+$\hide{min_{k\in {Neighbours(i)}}d(k,j)} minimum number of types in a path from type $j$ to type $i$). Now, define $d(i,U):=\min_{j \in U}d(i,j)$, where $U:=\{\,i: I_i^0>0\,\}$. Since we assumed
that every type $i$ is either in $U$ or is connected to a type in $U$,  $d(i,U)<M$ for all types $i$.

Let $\delta>0$ be a time such that for all types $i$ such that $d(i,U)=0$ (the initially infected types), \hide{and thus $I_i^{(0)}(0)=I_i^0>0$, }we have $I_i(t)>0$ for $t \in [0, \delta)$. Thus, proving Lemma \ref{lem:I=0_2} below will be equivalent to proving Lemma \ref{lem:I=0}, given an appropriate scaling of $\delta$.
\begin{Lemma}\label{lem:I=0_2}
For all $i$ and for all integers $r\geq 0$, if $d(i,U)\leq r$, then $I_i(t)>0$ for $t \in (\frac{r}{M}\delta, T)$.
\end{Lemma}
{\em Proof:}
By induction on r.

\emph{Base case:}
$r=0$. If $d(i,U)=0$, this means that the type is initially infected, and thus $I_i(t)>0$ for $t \in (0, T)$ by definition. Therefore the base case holds.

\emph{Induction step:}
Assume that the statement holds for  $r=0,\ldots,k$ and consider $r=k+1$. Since $(\frac{k+1}{M}\delta, T)\subset (\frac{k}{M}\delta, T)$, we need to examine types $i$ such that $d(i,U) = k+1$. In equation (\ref{Imm_Sys_NonRep}b) at $t=\frac{k+1}{M} \delta$, the first sum on the right involves terms like $I_j(\frac{k+1}{M}\delta)S_i(\frac{k+1}{M}\delta)$ where $j$ is a neighbor of $i$, while the second sum involves terms like $I_i(\frac{k+1}{M}\delta)u_j(\frac{k+1}{M}\delta)$. Since $d(i,U)= k+1$, there exist neighbours $j$ of $i$ such that $d(j,U)= k$, and therefore $I_j(t)>0$ for $t \in [\frac{k+1}{M}\delta,T)$ (by the induction hypothesis). Hence since $S_i^0>0$ and $\beta_{ji}>0$ ($i$ and $j$ being neighbours), for such $t$, $\dot{I}_{i}(t)> -I_{i}(t)\sum_{j=1}^M\pi_{j{i}} \bar\beta_{j{i}}R_j^0u_j(t)\geq-GI_{i}(t)$, where $G\geq 0$ is an upperbound on the sum (continuous functions are bounded on a closed and bounded interval). Thus $I_{i}(t)> I_{i}(\frac{k+1}{M}\delta) e^{-Kt}>0$, completing the proof for $r=k+1$.\quad\QED
\hide{
\section{\textsc{Proof of Lemma~\ref{baselemma}}}\label{appendix_lemma_6}
}
\hide{
We first prove that $\boldf(t,\cdot)$ is Lipschitz over $\BS$, where the Lipschitz constant  may be chosen uniformly as $t$ varies over $[0,T]$.
 Since the components of $Q_i(t)$ are  uniformly, absolutely bounded over $[0,T]$, it follows that the Euclidean norm $\|Q_i(t)\bx\|$ is uniformly bounded over $(t,\bx)\in[0,T]\times\BS$: that is to say, there exists $C<\infty$ such that $\sup_{[0,T]\times\BS} \|Q_i(t)\bx\| \leq C$. {Also, $\|{p_i} \| \leq H$ for some  $H<\infty$} as all its elements are bounded. Now, for each $t$, we may write
\begin{multline*}
  f_i(t,\bx) - f_i(t,\by)= \bx^T Q_i(t) \bx + {p_i}^T\bx - \by^T Q_i(t) \by - {p_i}^T\by
    \\= \bigl(\bx^T Q_i(t) \bx - \bx^T Q_i(t) \by\bigr)
      + \bigl(\bx^T Q_i(t) \by - \by^T Q_i(t) \by\bigr) \\+ {p_i}^T \bigl(\bx - \by\bigr)
    = \bigl(Q_i(t)\bx\bigr)^T (\bx - \by) + (\bx - \by)^T Q_i(t)\by + {p_i}^T \bigl(\bx - \by\bigr).
\end{multline*}
Taking absolute values of both sides, we obtain
\begin{multline*}
  |f_i(t,\bx) - f_i(t,\by)|\leq \bigl|\bigl(Q_i(t)\bx\bigr)^T (\bx - \by)\bigr|
      + \bigl|(\bx - \by)^T Q_i(t)\by\bigr| \\+ \bigl|{p_i}^T \bigl(\bx - \by\bigr)\bigr|
    \leq (\|Q_i(t)\bx\|+\|Q_i(t)\by\| +\|{p_i}\|)\cdot\|\bx-\by\| 
    \\\leq (2C+H) \|\bx-\by\| \qquad (t\in[0,T];\; \bx,\by\in\BS).
\end{multline*}
The first step follows by the triangle inequality, the second by two applications of the Cauchy-Schwarz inequality. Therefore
\begin{multline*}
  \|\boldf(t,\bx) - \boldf(t,\by)\| = \biggl(\sum_{i=1}^{3M}
      \bigl(f_i(t,\bx) - f_i(t,\by)\bigr)^2\biggr)^{\!1/2}
    \\\leq L\,\|\bx - \by\| \qquad (t\in[0,T];\; \bx,\by\in\BS),
\end{multline*}
and so $\boldf(t,\cdot)$ is Lipschitz over $\BS$ where the Lipschitz constant $L = (2C+H) \sqrt{3M}$ may be chosen uniformly as $t$ varies over $[0,T]$.

 We now show, using a standard process of Picard iteration,  that the solution $(t,\bF_0)\mapsto\bF(t,\bF_0)$ of~\eqref{Vectordf} is continuous, hence also uniformly continuous, over the compact set $[0,T]\times\BS$. Starting with any continuous function $\bFu0$ on $[0,T]\times\BS$, recursively form the Picard iterates
\begin{align*}
  \bFu{n}(t,\bx) = \bx + \int_0^t \boldf\bigl(v, \bFu{n-1}(v,\bx)\bigr)\,dv
    \qquad (n\geq1).
\end{align*}
[Subscripts are getting overloaded and so we introduce superscripts in a nonce notation to represent the iteration index (and not higher-order derivatives).]  We will show that  $\bigl\{\,\bFu{n}(t,\bx), n\geq1\,\bigr\}$  converges uniformly to a limit $\bF(t,\bx)$, where $\bF(t,\bx)$ is the unique solution to~\eqref{Vectordf} and is continuous over $[0,T]\times\BS$. This will complete the proof of the lemma.

 By induction we see that, for each $n$, $\bFu{n}$ is continuous, hence uniformly continuous and bounded, over the compact set $[0,T]\times\BS$. Write $K = \sup_{(t,\bx)\in[0,T]\times\BS} \bigl\|\bFu1(t,\bx) - \bFu0(t,\bx)\bigr\|$. The supremum over the compact set $[0,T]\times\BS$ is finite as the function difference inside the norm is continuous. As $\boldf$ is Lipschitz, $\bigl\|\bFu2(t,\bx) - \bFu1(t,\bx)\bigr\|
    \leq L\int_0^t \bigl\|\bFu1(v,\bx) - \bFu0(v,\bx)\bigr\|\,dv \leq KLt$.
So, by induction:
\begin{equation*}
  \bigl\|\bFu{n+1}(t,\bx) - \bFu{n}(t,\bx)\bigr\|\leq K\frac{(Lt)^n}{n!}.
\end{equation*}
By repeated use of the triangle inequality, it follows that
\begin{align*}
  \bigl\|\bFu{n+m}(t,\bx) - \bFu{n}(t,\bx)\bigr\|\leq K\sum_{k=n}^{n+m-1}\frac{(Lt)^k}{k!} 	\Longrightarrow\\
  \sup_{(t,\bx)\in[0,T]\times\BS}\bigl\|\bFu{n+m}(t,\bx) - \bFu{n}(t,\bx)\bigr\|
    \leq K\sum_{k=n}^{n+m-1}\frac{(LT)^k}{k!}\\\leq K e^{LT} \frac{(LT)^n}{n!}.
\end{align*}
The bound on the right decays to zero uniformly (as $n\to \infty$) and so $\bigl\{\,\bFu{n}(t,\bx), n\geq1\,\bigr\}$ is a Cauchy sequence converging uniformly to a limit $\bF(t,\bx)$. As the uniform limit of continuous functions is continuous, it follows that $\bF(t,\bx)$ is continuous over $[0,T]\times\BS$. Identifying $\bx$ with $\bF_0$, we next  show using  the usual Picard argument that this limit function $\bF(t,\bx)$ is a solution to~\eqref{Vectordf}.

 For any $0\leq t\leq T$, we have

\begin{align*}
  \biggl\| &\bFu{n+1}(t,\bx)
    - \biggl(\bx + \int_0^t \boldf\bigl(v,\bF(v,\bx)\bigr)\,dv\biggr)\biggr\|\allowdisplaybreaks\\&= \biggl\|\int_0^t \bigl[ \boldf\bigl(v,\bFu{n}(v,\bx)\bigr)
      - \boldf\bigl(v,\bF(v,\bx)\bigr)\bigr]\,dv\biggr\|\allowdisplaybreaks\\
    &\leq \int_0^t \bigl\| \boldf\bigl(v,\bFu{n}(v,\bx)\bigr)
      - \boldf\bigl(v,\bF(v,\bx)\bigr)\bigr\|\,dv
    \allowdisplaybreaks\\&\leq L\int_0^t \bigl\|\bFu{n}(v,\bx)  - \bF(v,\bx)\bigr\|\,dv\allowdisplaybreaks\\
    &\leq LT\sup_{0\leq v\leq T} \bigl\|\bFu{n}(v,\bx)  - \bF(v,\bx)\bigr\|.
\end{align*}
As $n\to\infty$, the term $\bFu{n+1}(t,\bx)$ on the left tends uniformly to the limit function $\bF(t,\bx)$ while the entire right-hand side tends to zero. It follows that
\begin{equation}\label{Vectordf2}\tag{\ref{Vectordf}$'$}
  \bF(t,\bx) = \bx + \int_0^t \boldf\bigl(v,\bF(v,\bx)\bigr)\,dv,
\end{equation}
which, by differentiation, is seen to be the same as~\eqref{Vectordf} with $\bx$ identified with $\bF_0$.

We next show that \eqref{Vectordf2} has a unique solution. Otherwise, let $\bF_1(t,\bx)$ and $\bF_2(t,\bx)$ constitute two distinct solutions of~\eqref{Vectordf2}. Let $J = J(\bx) = \sup_{0\leq t\leq T} \|\bF_1(t,\bx) - \bF_2(t,\bx)\|$. Then
\small
\begin{align*}
  \|\bF_1(t,\bx) &- \bF_2(t,\bx)\| = \biggl\|\int_0^t \bigl[ \boldf\bigl(v, \bF_1(v,\bx)\bigr)
      - \boldf\bigl(v, \bF_2(v,\bx)\bigr)\bigr]\,dv\biggr\|\\
    &\leq \int_0^t \bigl\|\boldf\bigl(v, \bF_1(v,\bx)\bigr)
      - \boldf\bigl(v, \bF_2(v,\bx)\bigr)\bigr\|\,dv\\&\overset{(\ast)}\leq L\int_0^t \|\bF_1(v,\bx) - \bF_2(v,\bx)\|\,dv
    \overset{(\ast\ast)}\leq JLt.
\end{align*}
\normalsize
Working iteratively by applying the bound~$(\ast\ast)$ to the integrand in the step~$(\ast)$, we see that
\small
\begin{align*}
  \|\bF_1(t,\bx) - \bF_2(t,\bx)\|
    &\leq L\int_0^t \|\bF_1(v,\bx) - \bF_2(v,\bx)\|\,dv
    \\&\leq JL^2\int_0^t v\,dv = \frac{JL^2t^2}2,
\end{align*}
\normalsize
whence, by induction, we obtain
$
  \|\bF_1(t,\bx) - \bF_2(t,\bx)\| \leq \frac{JL^nt^n}{n!}
$
for each $n$. The right-hand side tends to zero as $n\to\infty$ and so, by letting $n\to\infty$ on both sides, we see that $\bF_1(t,\bx) = \bF_2(t,\bx)$. Thus,  the system~\eqref{Vectordf} has a unique solution $\bF$.
\end{proof}}

\hide{\section{General Properties}\label{General Properties}
\subsection{Proof of Property \ref{property1}}
\begin{proof}
By contradiction. We assume $g(t_0) = L$ and $\dot{g}({t_0}^+)>0$. Then there exists a $\delta \in (0,t_1- t_0)$ such that  $\dot{g}(t)>0$ for $t \in(t_0, t_0+\delta)$. But $g(t_0+\delta)= g(t_0)+\int_{t_0}^{t_0+\delta} \! \dot{g}(x) \, \mathrm{d} x\hide{=L+\int_{t_0}^{t_0+\delta} \! \dot{g}(x) \, \mathrm{d} x}>L$, which is a contradiction. Thus the property holds. The proof for $g(t) > L$ has the signs interchanged.\hide{ The same argument can be applied from the left to get the left-sided counterparts.}
\end{proof}

\subsection{Proof of Property \ref{property2}}
\begin{proof}
We asume $g(a)=g(b)=L$. If  $\frac{dg}{dx}(a^+)>0$, there exists an $\epsilon>0$ such that $g(x)>L$ for all $x \in (a,a+\epsilon)$ and if $\frac{dg}{dx}(b^-)>0$ then there exists an $\alpha>0$ such that $g(x)<L$ for all $x \in (b-\alpha,b)$. Now $g(a+\frac{\epsilon}{2})>L$ and $g(b-\frac{\alpha}{2})<L$; thus, due to the continuity of $g(t)$, the intermediate value theorem states that there must exist a $y \in (a+\frac{\epsilon}{2},b-\frac{\alpha}{2})$ such that $g(y)=L$, which is in contradiction with the assumption that $g(x)\neq L$ for $x \in (a,b)$. The property follows.
\end{proof}}
\end{document}